\documentclass[11pt,letterpaper]{article}

\usepackage[margin=1in]{geometry}
\usepackage{theorem,latexsym,graphicx,amssymb}
\usepackage{amsmath,enumerate}
\usepackage{bbm}
\usepackage{wrapfig}
\usepackage{subfigure}
\usepackage{float}
\usepackage{psfig}
\usepackage{epsfig}
\usepackage{xspace}
\usepackage{paralist}
\usepackage{enumerate}
\usepackage{cases}
\usepackage{caption}
\usepackage{algorithm}
\usepackage[noend]{algpseudocode}
\usepackage{xcolor}
\usepackage{complexity}
\usepackage{multicol}
\usepackage{graphicx}

\usepackage{thm-restate}

\newenvironment{proof}{{\bf Proof:  }}{\hfill\rule{2mm}{2mm}\vspace*{5pt}}
\newenvironment{proofof}[1]{{\vspace*{5pt} \noindent\bf Proof of #1:  }}{\hfill\rule{2mm}{2mm}\vspace*{5pt}}
\numberwithin{figure}{section}
\numberwithin{equation}{section}
\newtheorem{definition}{Definition}[section]

\newtheorem{theorem}{Theorem}[section]
\newtheorem{lemma}{Lemma}[section]

\newtheorem{fact}{Fact}[section]

\newcommand{\fomp}{\textsf{Fully Online Matching}\xspace}
\newcommand{\fobmp}{\textsf{Fully Online Bipartite Matching}\xspace}
\newcommand{\obmp}{\textsf{Online Bipartite Matching}\xspace}

\newcommand{\ranking}{\textsf{Ranking}\xspace}

\newcommand{\eqdef}{\stackrel{\textrm{def}}{=}}
\newcommand{\expect}[2]{\operatorname{\mathbf E}_{#1}\left[#2\right]}
\newcommand{\vect}[1]{\ensuremath{\vec{#1}}}
\newcommand{\vecy}{\vect{y}}
\newcommand{\vecmv}[1][v]{\vect{y}_{\text{-}#1}}
\newcommand{\onev}{\ensuremath{\mathbbm{1}}}
\newcommand{\idr}[1]{\onev\left(#1\right)}

\newcommand*\samethanks[1][\value{footnote}]{\footnotemark[#1]}

\title{How to Match when All Vertices Arrive Online}

\author{Zhiyi Huang\thanks{Department of Computer Science, University of Hong Kong. {\texttt{\{zhiyi,nkang,zhtang,yhzhang2,xzhu2\}@cs.hku.hk}}}\and Ning Kang\samethanks \and Zhihao Gavin Tang\samethanks
	\and Xiaowei Wu\thanks{Department of Computing, Hong Kong Polytechnic University. {\texttt{csxwu@comp.polyu.edu.hk}.} The work was done when the author was a postdoc at the University of Hong Kong.}
	\and Yuhao Zhang\samethanks[1]
	\and Xue Zhu\samethanks[1]}
\date{}

\begin{document}

\begin{titlepage}
	\thispagestyle{empty}
	\maketitle
	
	\begin{abstract}
		We introduce a fully online model of maximum cardinality matching in which all vertices arrive online.
		On the arrival of a vertex, its incident edges to previously-arrived vertices are revealed.
		Each vertex has a deadline that is after all its neighbors' arrivals.
		If a vertex remains unmatched until its deadline, the algorithm must then irrevocably either match it to an unmatched neighbor, or leave it unmatched.
		The model generalizes the existing one-sided online model and is motivated by applications including ride-sharing platforms, real-estate agency, etc.
		
		We show that the Ranking algorithm by Karp et al.\ (STOC 1990) is $0.5211$-competitive in our fully online model for general graphs.
		Our analysis brings a novel charging mechanic into the randomized primal dual technique by Devanur et al.\ (SODA 2013), allowing a vertex other than the two endpoints of a matched edge to share the gain.
		To our knowledge, this is the first analysis of Ranking that beats $0.5$ on general graphs in an online matching problem, a first step towards solving the open problem by Karp et al.\ (STOC 1990) about the optimality of Ranking on general graphs.
		If the graph is bipartite, we show that the competitive ratio of Ranking is between $0.5541$ and $0.5671$.
		Finally, we prove that the fully online model is strictly harder than the previous model as no online algorithm can be $0.6317 < 1-\frac{1}{e}$-competitive in our model even for bipartite graphs.
	\end{abstract}
\end{titlepage}
\section{Introduction}

\obmp is a central problem in the area of online algorithms with a wide range of applications.
Consider a bipartite graph where the left-hand-side is known in advance, while vertices on the right-hand-side arrive online in an arbitrary order.
On the arrival of a vertex, its incident edges are revealed and the algorithm must irrevocably either match it to one of its unmatched neighbors or leave it unmatched.
Karp et al.~\cite{stoc/KarpVV90} introduced the \ranking algorithm, which picks at the beginning a random permutation over offline vertices, and matches each online vertex to the first unmatched neighbor according to the permutation.
Further, they proved that \ranking is $(1-\frac{1}{e})$-competitive
and the best possible among online algorithms.
The analysis of \ranking has been subsequently simplified in a series of papers~\cite{soda/GoelM08, sigact/BenjaminC08, soda/DevanurJK13}.
Further, it has been generalized to several extended settings, including the vertex-weighted case~\cite{soda/AggarwalGKM11}, the random arrival model~\cite{stoc/KarandeMT11, stoc/MahdianY11}, and the Adwords problem~\cite{jacm/MehtaSVV07, esa/BuchbinderJN07, stoc/DevanurJ12}.

However, all the above successful applications of \ranking crucially rely on the assumption that one side of the bipartite graph is known upfront. 
This assumption prevents us from applying the known positive results to some applications.
Here is an example:

\begin{center}
	\fcolorbox{lightgray}{lightgray}{
		\parbox{0.96\textwidth}{%
			\textbf{Example (Real Estate Agency).~} 
			During a typical day of a real estate agent in Hong Kong, both tenants and landlords drop by in an online fashion.
			Tenants specify what kinds of apartments they are looking for as well as the deadlines before which they need to move in\footnotemark. 
			Similarly, landlords list certain rules for tenant screening together with their deadlines.
			Tenants and landlords can be modeled as the vertices in a bipartite graph.
			There is an edge between a tenant-landlord pair if (1) they mutually satisfy each other's conditions, and (2) their time windows (between their respective arrivals and deadlines) overlap. 
			Real estate agents charge for each successful deal and, thus, seek to maximize the size of the bipartite matching. 
		}%
	}
\end{center}

\footnotetext{Tenants may also specify their earliest move-in dates. This is omitted in our model because, from the algorithmic point of view, it is equivalent to having each tenant arrive on the earliest move-in date. The same applies to landlords.}

This is clearly a bipartite matching problem with an online nature. 
However, it does not fit into the existing model in two fundamental ways.
First, vertices from both sides of the bipartite graphs arrive online.
Second, matching decision of each vertex is made at its deadline rather than its arrival.
There are many other applications with similar flavors such as job market intermediary and organ transplantation.

\paragraph{A Fully Online Model.}
Motivated by these applications, we formulate the following alternative online model of bipartite matching.
We call it {\fobmp} since vertices from both sides arrive online.
Let there be an underlying bipartite graph that is completely unknown to the algorithm at the beginning.
Each time step falls into one of the following two kinds:
\medskip
\begin{compactitem}
	\item Arrival of $v$: A vertex $v$ arrives; edges between $v$ and previously-arrived vertices are revealed. 
	\item Deadline of $v$: This is the last time a vertex $v$ can be matched (if it is not matched yet). 
\end{compactitem}
\medskip

The model further guarantees that all edges incident to a vertex are revealed before its deadline. 
Indeed, a tenant-landlord pair must have overlapping time windows in order to have an edge between them in the above example.
We further assume without loss of generality that algorithms are lazy in the sense that they only make decisions on the deadlines of the vertices. 
On the deadline of a vertex $v$, it might be the case that $v$ has already been matched to another vertex $u$ on $u$'s deadline.
Otherwise, the algorithm must irrevocably either match $v$ to one of its unmatched neighbors, or leave it unmatched. 

The previous one-sided online model is a special case in which all offline vertices arrive at the beginning and have deadlines at the end, and each online vertex has its deadline right after arrival.

Further, there are many applications for which the underlying graph is not necessarily bipartite.
Consider the following example:

\begin{center}
	\fcolorbox{lightgray}{lightgray}{
		\parbox{0.96\textwidth}{%
			\textbf{Example (Ride-sharing Platform).~} 
			DiDi is a major ride-sharing platform in China, handling tens of millions of rides on a daily basis.
			Requests are submitted to the platform in an online fashion.
			Each request is active in the system for a few minutes.
			The platform may match a pair of requests and serve them with the same taxi (or self-employed driver), provided that the pick-up locations and destinations are compatible, and their active time windows overlap.
			Requests can be modeled as vertices in a general graph and the compatibilities of pairs of requests can be modeled as edges.
		}
	}
\end{center}

Our model generalizes straightforwardly to general graphs by removing the bipartite assumption on the underlying graph. 
We refer to the generalization as \fomp.

It is easy to check that the na\"ive greedy algorithm that simply matches a vertex to an arbitrary unmatched neighbor remains to be $0.5$-competitive. 
Can we do better?

\subsection{Our Results and Techniques}

We consider a natural generalization of the \ranking algorithm that picks a random permutation over all vertices at the beginning, and matches each vertex (if unmatched at its deadline) to the first unmatched neighbor according to the permutation.
This algorithm can be implemented in our fully online model following the interpretation of \ranking by Devanur et al.~\cite{soda/DevanurJK13}:
On the arrival of $v \in V$, the rank of vertex $v$, denoted by $y_v$, is chosen uniformly at random from $[0, 1)$;
each vertex (if unmatched at its deadline) is matched to its unmatched neighbor with the highest, i.e., smallest, rank.
We show that the $\ranking$ algorithm is strictly better than $0.5$-competitive:

\begin{theorem}
	\label{th:general}
	\ranking is $0.5211$-competitive for \fomp.
\end{theorem}

\begin{theorem}
	\label{th:bipartite}
	\ranking is $0.5541$-competitive for \fobmp.
\end{theorem}

To our knowledge, our result for the \fomp problem is the first generalization of \ranking that achieves a competitive ratio strictly better than $0.5$ in an online matching model that allows general graphs, making the first step towards providing a positive answer to the open question of whether \ranking is optimal for general graphs by \cite{stoc/KarpVV90}.

\paragraph{Our Techniques (Bipartite Case).}
We build on the randomized primal dual technique introduced by Devanur et al.~\cite{soda/DevanurJK13}.
It can be viewed as a charging argument for sharing the gain of each matched edge between its two endpoints.
Whenever an edge $(u, v)$ is added to the matching, where $v$ is an offline vertex and $u$ is an online vertex, imagine a total gain of $1$ being shared between $u$ and $v$ based on the rank of the offline vertex $v$.
The higher the rank of $v$, the smaller share it gets.
For \obmp, Devanur et al.~\cite{soda/DevanurJK13} introduced a gain sharing method such that, for any edge $(u, v)$ and for any fixed ranks of offline vertices other than $v$, the expected gains of $u$ and $v$ (from all of their incident edges) combined is at least $1 - \frac{1}{e}$ over the randomness of $v$'s rank.
This implies the $1 - \frac{1}{e}$ competitive ratio.

Next, consider an edge $(u, v)$ in our model.
Suppose $u$ is the one with an earlier deadline. 
Since algorithms are lazy, the edge can only be added into the matching as a result of $u$'s decision at its deadline.
In this sense, $u$ plays a similar role as the online vertex and $v$ as the offline vertex in the analysis of Devanur et al.~\cite{soda/DevanurJK13}.
Hence, a natural attempt is to consider the expected gains of $u$ and $v$ combined in the charging argument over the randomness of $v$'s rank alone.

To explain why the above approach fails, we first introduce the notions of active and passive vertices. 
We say a vertex $u$ matches actively (or is active) if it is added to the matching by the algorithm at $u$'s deadline; other matched vertices are passive.
All previous analyses~\cite{stoc/KarpVV90, soda/AggarwalGKM11, soda/DevanurJK13, soda/ChanCWZ14, esa/AbolhassaniCCEH16} crucially rely on a structural property that whenever vertex $v$ is unmatched, its neighbor $u$ must be matched to some other vertex with rank higher than $v$. 
In our model, however, this holds only if $v$'s neighbor $u$ is active.

One may try to resolve this issue with a global amortized argument.
If we go over all the edges in the graph, it cannot be the case that the endpoint with an earlier deadline of every edge always matches passively.
After all, the numbers of active and passive vertices are equal.
Interestingly, we instantiate this intuition with a local amortized argument by taking expectation over the randomness of $u$'s rank as well.
Recall that $u$ is the vertex with earlier deadline and, thus, plays a similar role as the online vertex in the argument of Devanur et al.~\cite{soda/DevanurJK13}.
Taking expectation over $u$'s rank can be viewed as amortizing the case when $u$'s rank is low (active) and the case when $u$'s rank is high (passive).

\paragraph{Our Techniques (General Case).}
Moving from bipartite graphs to general graphs takes away another crucial structural property that the previous arguments rely on.
In a bipartite graph, if a vertex $u$ is matched by the \ranking algorithm for a realization of ranks while one of its neighbors $v$ is not, then $u$ remains matched no matter how the rank of $v$ changes.
In a general graph, however, it is possible that $u$ becomes unmatched when $v$ gets a higher rank.

We introduce a novel charging mechanic on top of the gain sharing rule used in the bipartite case.
After a matching has been chosen by \ranking, for each active vertex $w$, consider an alternative run of \ranking with the same ranks but with $w$ removed from the graph.
The difference between the two matchings will be an alternating path and, thus, at most one vertex $v$ would change from unmatched to matched in the absence of $w$.
We shall refer to such a vertex $v$ as the \emph{victim} of $w$.
Note that each active vertex has at most one victim, but an unmatched vertex could be the victim of many vertices.
If the victim of $w$ turns out to be its neighbor\footnote{This would induce an odd cycle and, thus, can only happen in non-bipartite graphs.}, our new charging mechanic will have $w$ send to $v$ a portion of $w$'s share from its incident edge in the matching, which we shall refer to as the \emph{compensation} from $w$ to $v$.
Further, we show that whenever the aforementioned structural property fails, namely, some vertex $u$ becomes unmatched when its unmatched neighbor $v$ gets a higher rank, we can always identify a unique neighbor $w$ of $v$ that sends a compensation to $v$ to remedy the loss in the charging argument.

Putting together the gain sharing mechanic from the bipartite case and the new mechanic of compensations, we can prove that for any edge $(u, v)$, the expected net gains of $u$ and $v$ combined is strictly greater than $0.5$ over the randomness of the ranks of both $u$ and $v$.

To our knowledge, this is the first charging mechanic that allows a vertex other than the two endpoints of a matched edge to get a share.
We believe this novel charging mechanic will find further applications in other matching problems that consider general graphs. 

\paragraph{Hardness Results.}

We complement our competitive analysis with two hardness results.
The first hardness applies to arbitrary online algorithms, showing a separation between the best competitive ratio in our fully online model and the optimal ratio of $1 - \frac{1}{e} \approx 0.6321$ in the existing one-sided online model. 
The second hardness focuses on the \ranking algorithm, certifying that our analysis for the bipartite case is close to the best possible.

\begin{theorem}
	\label{th:problem_hardness}
	No randomized algorithm can achieve a competitive ratio better than $0.6317$ for  \fobmp.
\end{theorem}

\begin{theorem}
	\label{th:ranking_hardness}
	\ranking is at most $0.5671$-competitive for \fobmp.
\end{theorem}

\subsection{Other Related Works}

An alternative generalization of \ranking to general graphs has been considered for the problem of oblivious matching~\cite{soda/ChanCWZ14, esa/AbolhassaniCCEH16}: Pick a permutation of vertices uniformly at random; then, go over the vertices one by one according to the permutation; for each unmatched vertex, match it to the first unmatched neighbor according to the same permutation.
Chan et al.~\cite{soda/ChanCWZ14} showed that it is a $0.523$-approximation algorithm, improving the previous $(\frac{1}{2}+\frac{1}{400000})$-approximation by a greedy algorithm~\cite{rsa/Aronson1995}.
Abolhassani et al.~\cite{esa/AbolhassaniCCEH16} improved the ratio to $0.526$ and analyzed the weighted case.
We stress that the alternative generalization is an offline algorithm because it needs to consider the vertices in random order, while our generalization is online.
Nevertheless, our result for general graphs can be viewed as a $0.5211$-approximation in the oblivious matching problem.
We believe the new algorithm and analysis in this paper, in particular, the new charging mechanic of compensations, will find further applications in oblivious matching and other matching problems that consider general graphs. 

Another online matching model in the literature considers online edge arrivals, upon which the algorithm must immediately decide whether to add the edge to the matching.
McGregor~\cite{approx/McGregor05} gave a deterministic $\frac{1}{3+2\sqrt{2}}\approx 0.1715$-competitive algorithm in the edge-weighted preemptive setting. 
This ratio is later shown to be tight for deterministic algorithms \cite{icalp/Varadaraja11}. 
Epstein et al.~\cite{stacs/EpsteinLSW13} designed a $\frac{1}{5.356}\approx 0.1867$-competitive randomized algorithm and proved a hardness of $\frac{1}{1+\ln 2}\approx 0.59$.
Chiplunkar et al.~\cite{esa/ChiplunkarTV15} considered a restricted setting where the input graph is an unweighted growing tree and gave a $\frac{15}{28}$-competitive algorithm. 
Finally, Buchbinder et al.~\cite{esa/BuchbinderST17} introduced an optimal $\frac{5}{9}$-competitive algorithm for unweighted forests.

Wang and Wong~\cite{icalp/WangW15} considered a more restrictive model of online bipartite matching with both sides of vertices arriving online:
A vertex can only actively match other vertices at its arrival; if it fails to match at its arrival, it may still get matched passively by other vertices later.
They showed a $0.526$-competitive algorithm for a fractional version of the problem.
We argue that the model in this paper better captures our aforementioned motivating applications.

\section{Preliminaries} \label{sec:preliminary}

We consider the standard competitive analysis against an oblivious adversary.
The competitive ratio of an algorithm is the ratio between the expected size of the matching by the algorithm over its own random bits to the size of the maximum matching of the underlying graph in hindsight.
The adversary must in advance choose an instance, i.e., the underlying graph as well as arrivals and deadlines of vertices, without observing the random bits used by the algorithm.
Otherwise, no algorithm can get any competitive ratio better than $0.5$.

\subsection{Ranking Algorithm and Some Basic Properties}

See Algorithm~\ref{alg:ranking} for a formal definition of \ranking in our model.
Let $M(\vecy)$ denote the matching produced when \ranking is run with $\vecy$ as the ranks.

\begin{algorithm}
	\caption{The \ranking Algorithm}
	\label{alg:ranking}
	\begin{algorithmic}
		\State (1) {a vertex $v$ arrives:}
		\State \phantom{(1)} pick $y_v \in [0,1)$ uniformly at random.
		\State (2) {a vertex $v$'s deadline reaches:}
		\State \phantom{(2)} \textbf{if} $v$ is unmatched,
		\State \phantom{(2) \textbf{if}} let $N(v)$ be the set of unmatched neighbors of $v$.
		\State \phantom{(2) \textbf{if}} \textbf{if} $N(v)=\emptyset$, \textbf{then} $v$ remains unmatched;
		\State \phantom{(2) \textbf{if}} \textbf{else} match $v$ to $\arg\min_{u\in N(v)} y_u$.
	\end{algorithmic}
\end{algorithm}

Recall the following definition of active/passive vertices. 
In the one-sided online model, only online vertices can be active and only offline vertices can be passive.
In our fully online model, however, vertices can in general be of either types depending on the random ranks of the vertices.

\begin{definition}[Active, Passive]
	For any edge $(u, v)$ added to the matching by \ranking at $u$'s deadline, we say that $u$ is active and $v$ is passive.
\end{definition}

The proofs of the following lemmas are deferred to Appendix~\ref{appendix:preli}.
The first lemma is a variant of the monotonicity property in previous works, incorporating the notions of active and passive vertices in our fully online model.

\begin{lemma}[Monotonicity]\label{lemma:monotonicity}
	For any rank vector $\vecy$ and any vertex $u$, we have
	\smallskip
	\begin{compactenum}
		\item if $u$ is active/unmatched, then $M(\vecy)$ remains the same when $y_u$ increases;
		\item if $u$ is passive, then $u$ remains passive when $y_u$ decreases.
	\end{compactenum}
\end{lemma}

Let $\vecmv[u] \in [0,1)^{V\setminus\{u\}}$ be the ranks of all vertices but $u$, i.e., $\vecmv[u]$ is obtained by removing the $u$-th entry in $\vecy$. 
Let $M(\vecmv[u])$ denote the matching produced by \ranking on $G-\{u\}$, i.e., the subgraph with vertex $u$ removed, with $\vecmv[u]$ as the ranks.

For ease of notation, for any $y\in[0,1]$, we use $y^\text{-}$ to denote a value that is arbitrarily close to, but smaller than $y$.
For example, our arguments consider functions discontinuous at $1$ and use $f(1^\text{-})$ to denote the limit of $f(x)$ as $x$ goes to $1$ from below.
We also consider the matching w.r.t.\ ranks $(y_u = \theta^\text{-}, \vecmv[u])$ to avoid confusions in marginal cases where ranks $(y_u = \theta, \vecmv[u])$ need tie-breaking.

By Lemma~\ref{lemma:monotonicity}, we can uniquely define the following marginal rank for every vertex.

\begin{definition}[Marginal Rank]
	For any $u$ and any ranks $\vecmv[u]$ of other vertices, the {marginal rank} $\theta$ of $u$ with respect to $\vecmv[u]$ is the largest value such that $u$ is passive in $M(y_u=\theta^\text{-}, \vecmv[u])$.
\end{definition}

Note that a vertex may still match another vertex (actively) when its rank is below the marginal rank in our fully online model.
Nevertheless, it is consistent with the previous definition in the one-sided online model that concerns offline vertices, which cannot match actively.

\begin{lemma}[Unmatched Neighbor]\label{lemma:unmatched_neighbor}
	Suppose $v$ has marginal rank $\theta < 1$ with respect to $\vecmv$.
	Then, for any neighbor $u$ of $v$ that has an earlier deadline than $v$, and for any rank vector $(y_v=y, \vecmv)$ with $y \in [\theta,1)$, $u$ either is passive, or actively matches a vertex with rank at most $\theta$.
\end{lemma}

It is well known that removing a matched vertex from the graph results in an alternating path in the matching produced by \ranking.
The next lemma provides a more fine-grained characterization.

\begin{lemma}[Alternating Path]\label{lemma:alternating_path}
	If $u$ is matched in $\vecy$, then the symmetric difference between the matchings $M(\vecy)$ and $M(\vecmv[u])$ is an alternating path $(u_0, u_1,\ldots,u_l)$ with $u_0=u$ such that
	\medskip
	\begin{compactenum}
		\item for all even $i<l$, we have $(u_i,u_{i+1})\in M(\vecy)$; for all odd $i<l$, we have $(u_i,u_{i+1})\in M(\vecmv[u])$;
		\item from $M(\vecy)$ to $M(\vecmv[u])$, vertices $\{u_1,u_3,\ldots\}$ get worse, vertices $\{u_2,u_4,\ldots\}$ get better.
	\end{compactenum}
	\medskip
	Here, passive is better than active, which is in turns better than unmatched. 
	Conditioned on being passive, matching to a vertex with earlier deadline is better.
	Conditioned on being active, matching to a vertex with higher rank is better.
\end{lemma}

\subsection{Randomized Dual Fitting}

Consider the following linear program relaxation of the matching problem and its dual.
\begin{align*}
	\max: \quad & \textstyle \sum_{(u,v)\in E} x_{uv} && \qquad\qquad & \min: \quad & \textstyle\sum_{u \in V} \alpha_u\\
	\text{s.t.} \quad & \textstyle \sum_{v:(u,v)\in E} x_{uv} \leq 1 && \forall u\in V & \text{s.t.} \quad & \alpha_u + \alpha_v \geq 1 && \forall (u,v)\in E \\
	& x_{uv} \geq 0 && \forall (u,v)\in E & & \alpha_u \geq 0 && \forall u \in V
\end{align*}

It is known that the above linear program relaxation is integral for bipartite graphs, but it has a large integrality gap for general graphs (e.g., a complete graph of $3$ vertices).
Interestingly, this relaxation is sufficient for proving our positive results, even for general graphs.

Our approach builds on the randomized primal dual technique by Devanur et al.~\cite{soda/DevanurJK13}.
We believe it is more appropriate to call our analysis (for general graphs) randomized dual fitting, however, because it relies on an extra phase of adjustments to the dual variables at the end that requires full knowledge of the instance.

\paragraph{Randomized Dual Fitting.}
We set the primal variables according to the matching by \ranking, which ensures primal feasibility, and set the dual variables such that the dual objective equals the primal objective.
The dual assignment can be viewed as splitting the gain of $1$ of every matched edge among the vertices; the dual variable $\alpha_v$ for every vertex $v$ is equal to the total share it gets from all matched edges.
Given primal feasibility and equal objectives, the usual primal dual and dual fitting techniques would further seek to show approximate dual feasibility, namely, $\alpha_u + \alpha_v \ge F$ for every edge $(u, v)$ where $F$ is the target competitive ratio.
This is where the usual techniques fail and the smart insight by Devanur et al.~\cite{soda/DevanurJK13} comes to help.
Due to the intrinsic randomness of \ranking, the above primal and dual assignments are themselves random variables.
Devanur et al.~\cite{soda/DevanurJK13} observe that it suffices to have approximate dual feasibility in expectation.
For completeness, we formulate this insight as the following lemma and include a proof in Appendix~\ref{appendix:preli}.

\begin{lemma}\label{lemma:dual_fitting}
	\ranking is $F$-competitive if we can set (non-negative) dual variables such that 
	\smallskip
	\begin{compactenum}
		\item $\sum_{(u,v)\in E} x_{uv} = \sum_{u \in V} \alpha_u$; and
		\smallskip
		\item $\E_{\vecy}[\alpha_u+\alpha_v] \geq F$ for all $(u,v)\in E$.
	\end{compactenum}
\end{lemma}

\section{Bipartite Graphs: A Warm-up}
\label{sec:bipartite}

\paragraph{Dual Assignment.}
We adopt the dual assignment by Devanur et al.~\cite{soda/DevanurJK13} and share the gain of each matched edge between its two endpoints as follows:
\begin{itemize}
	\item \emph{Gain Sharing:~} Whenever an edge $(u, v)$ is added to the matching with $u$ active and $v$ passive, let $\alpha_u = 1 - g(y_v)$ and $\alpha_v = g(y_v)$. 
	Here, $g : [0,1] \rightarrow [0,1]$ is non-decreasing with $g(1) = 1$.
\end{itemize}


\paragraph{Randomized Primal Dual Analysis.}
The previous analysis of \ranking for \obmp relies on a structural property that for any edge $(u, v)$ and any ranks $\vecmv[v]$, $u$ matches a vertex with rank no larger than $v$'s marginal rank regardless of $v$'s rank (e.g. Lemma 2.3 of~\cite{soda/DevanurJK13}).
However, in our fully online setting, the same property holds only when $u$ is active.
By introducing the notions of passive and active vertices, we show the following weaker version of the property.
It complements the basic property when $y_u$ is larger than the marginal rank  (Lemma~\ref{lemma:unmatched_neighbor}).

\begin{lemma}\label{lemma:friendly_partner_bipartite}
	Suppose $v$ has marginal rank $\theta < 1$ with respect to $\vecmv$.
	Then, for any neighbor $u$ of $v$ that has an earlier deadline than $v$, and for any rank vector $(y_v = y, \vecmv)$ with $y \in [0, \theta)$, $u$ either is passive, or matches actively to a vertex with rank at most $\theta$. 
\end{lemma}
\begin{proof}
	We consider the matchings in 3 sets of ranks $\vecy = (y_v = y, \vecmv)$, $\vecmv[v]$ and $\vec{y}_{\theta} = (y_v=\theta, \vecmv)$. First, we show that $u$ matches the same neighbor in $M(\vecy_{\theta})$ and $M(\vecmv)$. Since $v$ is unmatched or active in $M(\vecy_{\theta})$, removing $v$ cannot affect vertices with earlier deadlines. In particular, $u$ would match the same neighbor.
	
	Consider the alternating path from $M(\vecy)$ to $M(\vecmv)$. 
	If $u$ is not in the alternating path, then $u$ matches the same neighbor in all $M(\vecy)$, $M(\vecmv[v])$ and $M(\vecy_{\theta})$. 
	Otherwise, $u$ appears in the alternating path with an odd distance from $v$ since the graph is bipartite.
	Hence, by Lemma~\ref{lemma:alternating_path}, $u$ is better in $M(\vecy)$ than in $M(\vecmv)$ and, thus, is better than in $M(\vecy_{\theta})$.
	In both cases, $u$ is passive or actively matches a vertex with rank $\le \theta$ in $M(\vecy)$, since this holds for $u$ in $M(\vecy_{\theta})$ (by Lemma~\ref{lemma:unmatched_neighbor}).
\end{proof}

Recall that for any edge $(u, v)$ we will consider the expected gain of $\alpha_u$ and $\alpha_v$ combined over the randomness of the ranks of both $u$ and $v$.
First, let us fix the rank of $u$, the vertex with an earlier deadline, and consider the expected gain over the randomness of $v$'s rank alone.

\begin{lemma}\label{lemma:bipartite_main}
	For any neighbor $u$ of $v$ that has an earlier deadline than $v$, and for any $\vecmv$, we have
	\[
	\textstyle
	\expect{y_v}{\alpha_u + \alpha_v} \ge f(y_u) \eqdef \min_{\theta\in[0,1]} \left\{ \int_0^\theta g(y_v)dy_v + \min\{1-g(\theta),g(y_u)\} \right\}.
	\]
\end{lemma}
\begin{proof}
	Let $\theta$ be the marginal rank of $v$ with respect to ranks $\vecmv$. By definition, $v$ is passive and gets $g(y_v)$ when $y_v < \theta$, i.e. $\expect{y_v}{\alpha_v \cdot \idr{y_v < \theta}} = \int_{0}^{\theta} g(y_v) dy_v$. By Lemma~\ref{lemma:unmatched_neighbor} and~\ref{lemma:friendly_partner_bipartite}, $\expect{y_v}{\alpha_u} \ge \min\{1-g(\theta), g(y_u)\}$. Adding them together and taking the minimum over all possible $\theta$'s concludes the statement.
\end{proof}

It is worthwhile to make a comparison with a similar claim in the previous analysis  by Devanur et al.~\cite{soda/DevanurJK13} for \obmp:
\[
\textstyle
\expect{y_v}{\alpha_u + \alpha_v} \ge \min_{\theta\in[0,1]} \left\{ \int_0^\theta g(y_v)dy_v + 1-g(\theta) \right\} ~,
\]
where $u$ is an online vertex and $v$ is an offline vertex.
As we have discussed in the introduction, for every edge in our model, the endpoint with an earlier deadline plays a similar role as the online vertex in the previous one-sided online model since the edge can only be added to the matching as a result of this endpoint's matching decision.
In this sense, the bounds are indeed very similar, except for the last term, where the previous bound simply has $1 - g(\theta)$ while our bound has the smaller of $1 - g(\theta)$ and $g(y_u)$.

We interpret Lemma~\ref{lemma:bipartite_main} as follows.
It recovers the previous bound when the rank of $u$ is large ($1 - g(\theta) \le g(y_u)$), which roughly corresponds to the case when $u$ is active (or unmatched) and the previous structural property holds.
When the rank of $u$ is small ($1 - g(\theta) > g(y_u)$), which roughly corresponds to the case when $u$ is passive, it still provides some weaker lower bound on the expected gains of the two endpoints.
The weaker bound, however, is at most $0.5$ in the worst case:
the RHS becomes $\min \{1 - g(0), g(0) \} \le 0.5$ for $\theta = y_u = 0$.
Hence, it is crucial that we take expectation over the randomness of $u$'s rank as well, effectively amortizing the cases when $u$ is active and when it is passive.
This idea carries over to general graphs.

\bigskip

\begin{proofof}{Theorem~\ref{th:bipartite}}
	Let $g(x) = e^{x-1}$, where $x\in[0,1]$. For all $(u,v) \in E$, by Lemma~\ref{lemma:bipartite_main},
	\[
	\textstyle
	\expect{\vecy}{\alpha_u+\alpha_v} = \expect{\vecmv}{\expect{y_v}{\alpha_u+\alpha_v}} \ge \expect{\vecmv}{f(y_u)} = \int_{0}^{1} f(y_u) dy_u.
	\]
	
	Observe that 
	\[
	\textstyle
	\int_0^\theta g(z)dz + 1-g(\theta) = e^{\theta-1}-\frac{1}{e}+1-e^{\theta-1}=1-\frac{1}{e}
	\]
	and $\int_0^\theta g(z)dz + g(x) \geq g(x)$ for all $\theta\in[0,1]$.
	We have $f(x)\geq \min\{g(x),1-\frac{1}{e}\}$, which implies (let $\theta = \ln(e-1)$ s.t. $g(\theta) = 1-\frac{1}{e}$)
	\[
	\textstyle
	\int_{0}^{1} f(x) dx \geq \int_{0}^{\theta}e^{x-1}dx + (1-\theta)(1-\frac{1}{e}) = \frac{e-2}{e}+(1-\ln(e-1))(1-\frac{1}{e}) \approx 0.55418.
	\]
	By Lemma~\ref{lemma:dual_fitting}, we conclude that \ranking is at least $0.5541$-competitive.
\end{proofof}

We are aware of a different function $g(y) = \min\{1,e^{y-1}+0.0128\}$ that gives a (very slightly) better competitive ratio $0.5547$.
For convenience of presentation we only fix a simple form here.

\section{General Graphs: An Overview} \label{sec:general}

\paragraph{Dual Assignment.}
Moving from bipartite graphs to general graphs, even the weaker version of the structural property, i.e., Lemma~\ref{lemma:friendly_partner_bipartite}, ceases to hold.
Consider an edge $(u, v)$ with $u$'s deadline being earlier.
It is possible that decreasing $y_v$ leads to a change of $u$'s status from matched to unmatched in a non-bipartite graph.
As a result, the simple gain sharing rule in the previous analysis on the bipartite case no longer gives any bound strictly better than $0.5$.

To handle general graphs, we design a novel charging mechanic on top of the gain sharing rule between the endpoints of matched edges.
First, we introduce the following notion of \emph{victim}.

\begin{definition}[Victim]
	For any ranks $\vecy$ and any active vertex $w$, $v$ is $w$'s victim if
	\smallskip
	\begin{compactitem}
		\item $v$ is an unmatched neighbor of $w$;
		\item $v$ is matched in $M(\vecmv[w])$.
	\end{compactitem}
\end{definition}

Observe that removing $w$ results in an alternating path (Lemma~\ref{lemma:alternating_path}) and, thus, at most one vertex changes from unmatched to matched. 
Hence, each active vertex has at most one victim.

Consider the following two-step approach for computing a dual assignment:
\medskip
\begin{compactitem}
	\item \emph{Gain Sharing:~} Whenever an edge $(u, v)$ is added to the matching with $u$ active and $v$ passive, let $\alpha_u = 1 - g(y_v)$ and $\alpha_v = g(y_v)$. 
	Here, $g : [0,1] \rightarrow [0,1]$ is non-decreasing with $g(1) = 1$.
	\smallskip
	\item \emph{Compensation:~} For every active vertex $u$ that has a victim $z$, suppose $u$ is matched to $v$.
	Decrease $\alpha_u$ and increase $\alpha_z$ by the same amount $h(y_v)$, where $h: [0,1] \rightarrow [0,1]$ is non-decreasing in $[0,1)$, $h(y) / y$ is non-increasing, $h(1) = 0$ and $1 - g(y) - h(y) \ge 0$ for all $y$.
	%
\end{compactitem}
\medskip
Note that the second step, in particular, identifying the victims of active vertices, can only be done after the entire instance has been revealed.


Each matched vertex will gain only from its incident matched edge.
If it is further active and has a victim, it needs to send a compensation to the victim.
Further, the active vertex can always afford the compensation from its gain since $1-g(y)-h(y)\ge 0$ for all $y\in[0,1)$.
The monotonicity of $h(y) / y$ is for technical reasons in the analysis.
Finally, note that an unmatched vertex may receive compensations from any number of active vertices.

\paragraph{Randomized Dual Fitting Analysis.}
The main technical lemma is to establish a lower bound for $\expect{y_v}{\alpha_u + \alpha_v}$, as we have done in Lemma~\ref{lemma:bipartite_main} for bipartite graphs.
Due to space constraint, we present the analysis for a special case with following assumptions ($\theta$ is the marginal rank of $v$):
%
\medskip
\begin{compactitem}
	\item $v$ is unmatched in $M(y_v=y, \vecmv)$ for all $y \ge \theta$;  
	\item $u$ actively matches the same vertex $z$ with rank $y_z=\tau>\theta$ in $M(y_v=y, \vecmv)$ for all $y <\theta$. 
\end{compactitem}
\medskip
%
In other words, any rank of $v$ higher than its marginal rank leads to the same (worse) situation for $u$, i.e. matching a vertex with rank $\tau \in (\theta,1)$.
See Figure~\ref{fig:general-illustration} for an illustrative example.

\begin{figure}[t]
	\centering
	\subfigure[when $y_v = \theta^\text{-}$]{\centering \includegraphics[width = 0.21\textwidth]{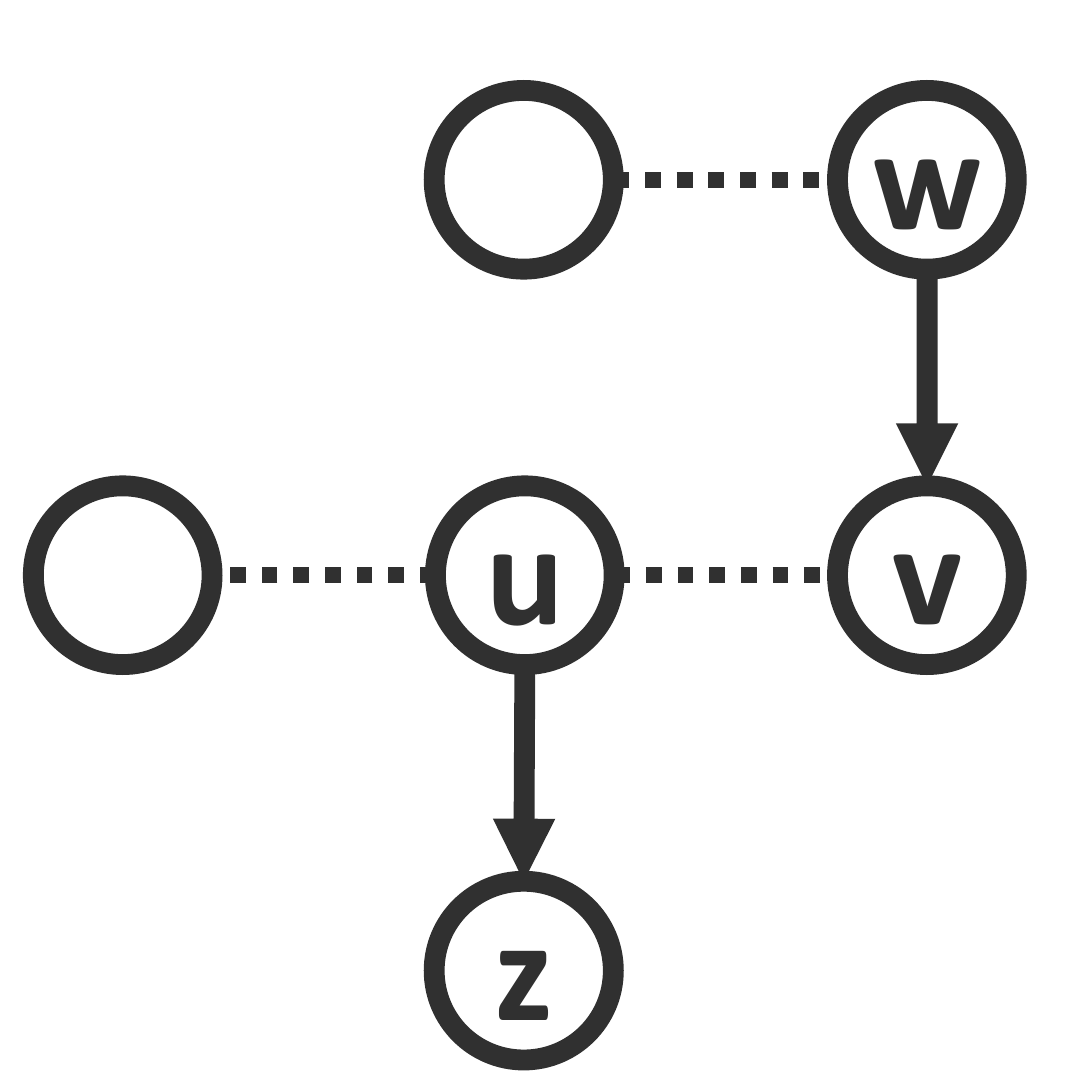}}
	\quad
	\subfigure[when $y_v \in [\theta,\tau)$]{\centering\includegraphics[width = 0.25\textwidth]{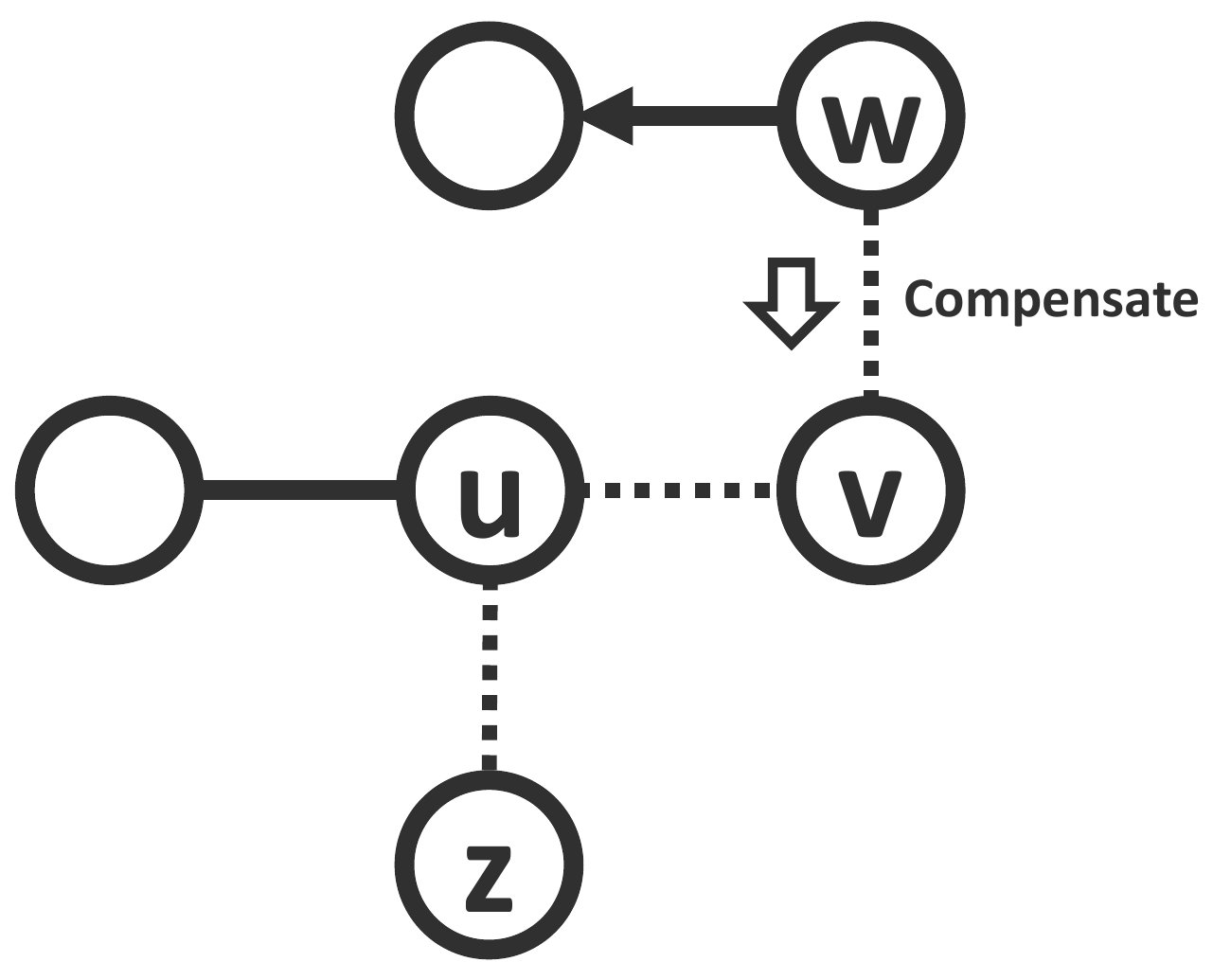}}
	\quad
	\subfigure[the symmetric difference]{\centering\includegraphics[width = 0.34\textwidth]{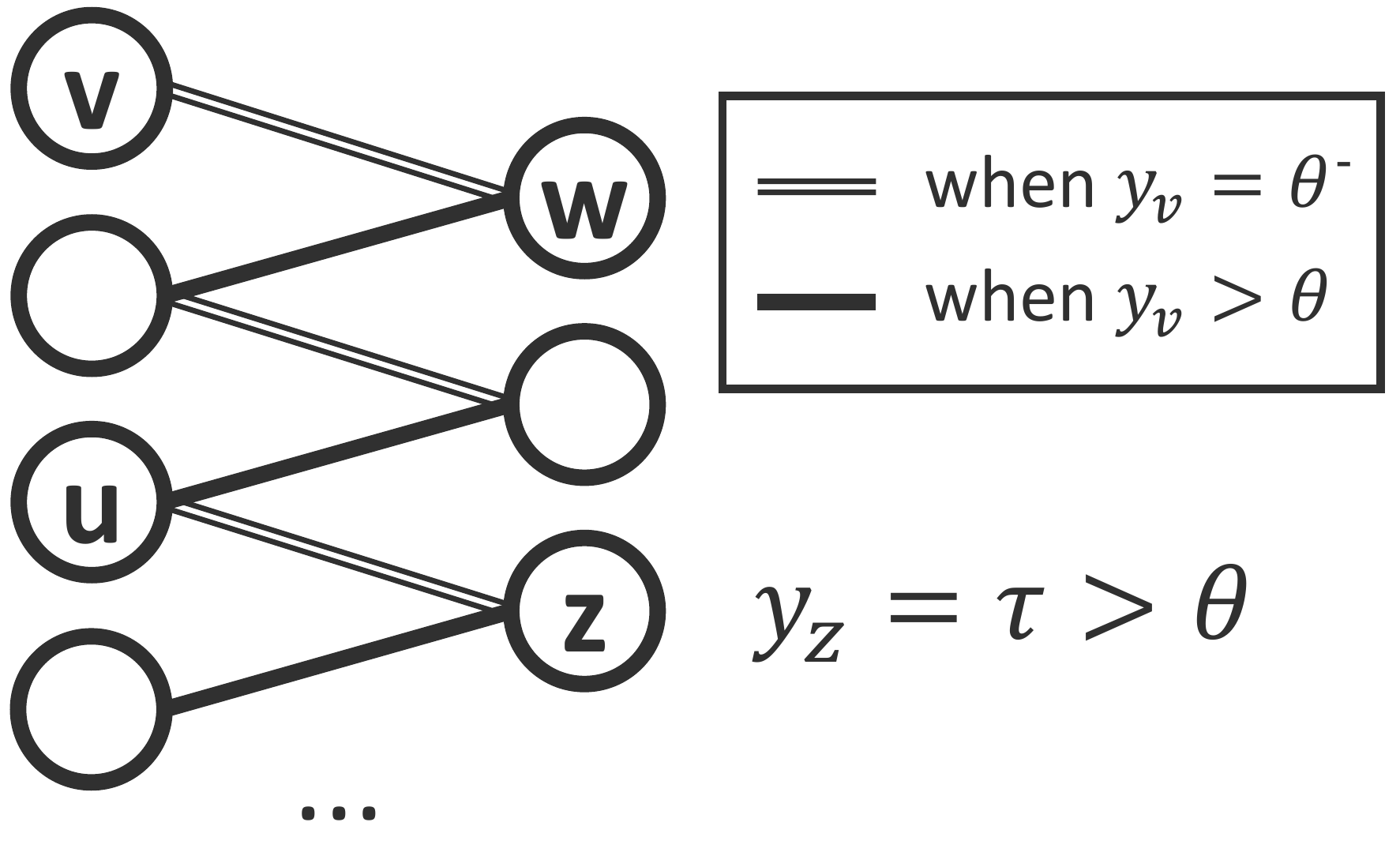}}
	\caption{We use solid line to represent an edge in the matching, where the direction (if any) is from active vertex to passive vertex.
		(a) when $v$ is higher than its marginal rank, $u$ matches a vertex with rank $y_z = \tau > \theta$;
		(b) $v$ is unmatched and compensated by $w$ when it is lower than its marginal rank; $u$ either is passive, or matches actively some vertex with rank higher than $\tau$ since $v$ is unmatched;
		(c) the symmetric difference between the two matchings: an alternating path triggered when $y_v$ increases to be larger than its marginal rank.}
	\label{fig:general-illustration}
\end{figure}

In general, we need to also consider the case that $v$ is active when its rank is lower than the marginal $\theta$, and the possibility that $u$'s matching status may change multiple times as the rank of $v$ gets higher.
See Appendix~\ref{subsec:general_main} for the analysis without the simplifying assumptions.

Subject to the above simplifying assumptions, we show the following:

\begin{lemma} \label{lemma:general_main_sketch}
	For any neighbor $u$ of $v$ that has an earlier deadline than $v$, and for any $\vecmv$, we have
	\begin{align*}
		\expect{y_v}{\alpha_v + \alpha_u} \ge & \textstyle \int_{0}^{\theta} g(y_v) dy_v + (\tau - \theta) \cdot h(\theta) + \theta \cdot (1-g(\tau)-h(\tau)) \\
		& + (1-\theta) \cdot \min \big\{ g(y_u) ,1-g(\theta)-h(\theta) \big\} ~.
	\end{align*}
\end{lemma}

Suppose $w$ matches actively to $v$ in $M(y_v=\theta^{\text{-}},\vecmv)$ (refer to Figure~\ref{fig:general-illustration}(a)), that is, it is the first vertex after $v$ in the alternating path when $v$'s rank moves below its marginal rank (refer to Figure~\ref{fig:general-illustration}(c)). 
We show in the following lemma that $v$ receives a compensation from $w$ whenever its rank $y_v$ is between $\theta$ and $\tau$ (refer to Figure~\ref{fig:general-illustration}(b)).

\begin{lemma} \label{lem:compensate_special}
	For any $y \in [\theta, \tau)$, $v$ is the victim of $w$ in $M(y_v = y, \vecmv)$. 
\end{lemma}
\begin{proof}
	Let $\vecy_1 = (y_v = y, \vecmv)$, where $y \in [\theta, \tau)$.
	By our assumption, $v$ is unmatched and, thus, is an unmatched neighbor of $w$ in $M(\vecy_1)$. 
	To prove that $v$ is the victim of $w$, we need to show that (1) $w$ is active in $M(\vecy_1)$ and (2) $v$ becomes matched when we remove $w$ from the graph.
	
	Consider $\vecy_2 = (y_v = \theta^{\text{-}}, \vecmv)$.
	By our assumptions, $v$ is passively matched to $w$ and $u$ actively matches $z$ with $y_z = \tau$ in $M(\vecy_2)$.
	For this to happen, $w$ must have an earlier deadline than $v$ and none of matching decisions before $w$'s deadline pick $w$ or $v$.
	Then, lowering $v$'s rank would not affect these decisions before $w$'s deadline and, thus, $w$ must also be active in $M(\vecy_1)$. 
	
	Finally, consider what happens when $w$ is removed from the graph.
	It triggers a portion of the alternating path (i.e., Figure~\ref{fig:general-illustration}(c)), the symmetric difference between $M(\vecy_1)$ and $M(\vecy_2)$.
	The portion starts from $w$ (exclusive) and ends the first time when $v$ becomes relevant, i.e., a vertex in the alternating path decides to pick $v$ instead of the next vertex in the path.
	Further, we know for sure that $v$ will be relevant at some point because otherwise $u$ is in the path and the next vertex $z$ has rank $\tau > y$.
	Therefore, $v$ must be matched when $w$ is removed from the graph.
\end{proof}


The next two lemmas give lower bounds on the expected gain of $\alpha_v$ and $\alpha_u$, respectively, over the randomness of $v$'s rank alone.

\begin{lemma} 	\label{lem:alphav_special}
	$\expect{y_v}{\alpha_v} \ge \int_{0}^{\theta} g(y_v) dy_v + (\tau - \theta) \cdot h(\theta) $.
\end{lemma}
\begin{proof}
	By definition, when $y_v < \theta$, $v$ is passive and hence $\alpha_v = g(y_v)$.
	
	Since $w$ matches $v$ actively in $M(y_v=\theta^{\text{-}}, \vecmv)$ but not in $M(y_v = \theta, \vecmv)$, we know that $w$ must match a vertex with rank $\theta$ in $M(y_v \ge \theta, \vecmv)$. For $y_v \in [\theta, \tau)$, Lemma~\ref{lem:compensate_special} implies that $v$ is the victim of $w$ and, thus, $v$ receives a compensation $h(\theta)$ from $w$.
	To sum up, we have
	\begin{equation*}
		\textstyle
		\expect{y_v}{\alpha_v} \ge \expect{y_v}{\alpha_v \cdot \idr{y_v < \theta}} + \expect{y_v}{\alpha_v \cdot\idr{y_v \in [\theta, \tau)}} \ge \int_{0}^{\theta} g(y_v) dy_v + (\tau - \theta) h(\theta),
	\end{equation*}
	as claimed.
\end{proof}

\begin{lemma}
	\label{lem:alphau_special}
	$\expect{y_v}{\alpha_u}	\ge \theta \cdot (1-g(\tau)-h(\tau)) + (1-\theta) \cdot \min \{g(y_u) ,1-g(\theta)-h(\theta)\}$.
\end{lemma}
\begin{proof}
	By assumption, $u$ actively matches vertex $z$ with rank $y_z = \tau$ when $y_v < \theta$.
	Thus, $u$ gains $1-g(\tau)$ during the gain sharing phase and gives away $h(\tau)$ to its victim (if any).
	Integrating $y_v$ from $0$ to $\theta$ gives the first term on the RHS.
	
	For $y_v\ge \theta$, $u$ either is passive, or actively matches a vertex with rank at most $\theta$.
	In the first case, we have $\alpha_u = g(y_u)$.
	In the second case, we have $\alpha_u \ge 1-g(\theta)-h(\theta)$, by the monotonicity of $g,h$.
	Integrating $y_v$ from $\theta$ to $1$ gives the second term on the RHS.
\end{proof}

Summing up the inequalities in Lemma~\ref{lem:alphav_special} and \ref{lem:alphau_special} proves Lemma~\ref{lemma:general_main_sketch}.

Comparing the bounds of Lemma~\ref{lem:alphav_special} and Lemma~\ref{lem:alphau_special}, the parameter $\tau$ presents a trade-off between the expected gains $\alpha_u$ and $\alpha_v$ of the two vertices.
The larger $\tau$ is, the less $u$ gets when $v$ is above its marginal rank, e.g., $1-g(\tau)-h(\tau)$, and the more $v$ gets as compensations when it is below its marginal rank, e.g., $(\tau-\theta) \cdot h(\tau)$; and vice versa.

\paragraph{Charging Functions.}
If Lemma~\ref{lemma:general_main_sketch} holds unconditionally, it remains to show that there exists functions $g,h$ (with desired properties) and constant $F > 0.5$ such that
\begin{align*}
	\int_0^1 \min_{0\leq \theta < \tau < 1} \bigg\{ & \int_{0}^{\theta} g(y_v) dy_v + (\tau - \theta) h(\theta) + \theta \cdot (1-g(\tau)-h(\tau)) \\
	& + (1-\theta) \cdot \min \big\{ g(y_u) ,1-g(\theta)-h(\theta) \big\} \bigg\} dy_u \geq F ~,
\end{align*}
and to apply Lemma~\ref{lemma:dual_fitting} to conclude that \ranking is $F$-competitive.
The unconditional version of Lemma~\ref{lemma:general_main_sketch} turns out to give a more complicated bound due to the considerations of other cases.
Nevertheless, we can use a linear program to optimize the ratio over fine-grained discretized versions of $g$ and $h$.
To give a rigorous proof, which is deferred to Appendix~\ref{subsec:ratio}, we approximate the solutions of the linear program with piecewise-linear $g$ and $h$ (with two segments).
We conclude with our choice of $g$ and $h$ that \ranking is $0.5211$-competitive.

\newpage

{
	\bibliography{matching}
	\bibliographystyle{alpha}
}

\newpage

\appendix
\section{Missing Proofs in Section~\ref{sec:preliminary}}\label{appendix:preli}

\begin{proofof}{Lemma~\ref{lemma:monotonicity}}
	For the first statement, since $u$ is active or unmatched, we know that	for each neighbor $v$ of $u$ with an earlier deadline than $u$,
	$v$ does not match $u$ in $M(\vecy)$ at their deadlines. 
	Hence when $y_u$ increases, they would make the same decision.
	In other words, when $u$'s deadline reaches, the partial matching produced is the same as before.
	As a consequence, the eventual matching produced would be identical, as $u$ will actively match the same vertex as in $M(\vecy)$.
	
	The second statement is implied by the first statement.
	Suppose otherwise, e.g., $u$ is active or unmatched when $y_u$ is decreased from $y$ to some $y' < y$.
	Then we know that by increasing $y_u$ from $y'$ to $y$, $u$ becomes passive, which violates the first statement.
\end{proofof}

\begin{proofof}{Lemma~\ref{lemma:unmatched_neighbor}}
	Consider the matching $M(y_v=\theta, \vecmv)$. By definition, $v$ is either active or unmatched.
	Hence, at $u$'s deadline, which is earlier than $v$'s deadline, $v$ is unmatched.
	Consequently, $u$ either is passive or matches actively to some vertex $w$ with $y_w \le y_v = \theta$.
	By Lemma~\ref{lemma:monotonicity}, there is no change in the matching when we increases $y_u$, which concludes the proof.
\end{proofof}

\begin{proofof}{Lemma~\ref{lemma:alternating_path}}
	We prove the lemma by mathematical induction on $n$, the total number of vertices.
	For the base case when $n = 2$, the symmetric difference is a single edge $(u,u_1)$ and the second statement holds since $u_1$ is matched in $M(\vecy)$ and unmatched in $M(\vecmv[u])$.
	
	Suppose the lemma holds for $1,2,\ldots,n-1$ and we consider the case when $|V| = n$.
	
	Let $u_1$ be matched to $u$ in $M(\vecy)$.
	Observe that if we remove both $u$ and $u_1$ from $G$ (let $\vec{y'} \in [0,1]^{V\setminus\{u,u_1\}}$ be the resulting vector), then we have $M(\vecy) = M(\vec{y'})\cup \{ (u,u_1) \}$.
	
	If $u_1$ is unmatched in $M(\vecmv[u])$, then we have $M(\vecmv[u]) = M(\vec{y'})$ and the lemma holds by induction hypothesis.
	Now suppose $u_1$ is matched in $M(\vecmv[u])$.
	
	By definition $\vect{y'}$ is obtained by removing $u_1$ (which is matched in $\vecmv[u]$) from $\vecmv[u]$.
	By induction hypothesis, the symmetric difference between $M(\vecmv[u])$ and $M(\vect{y'})$ is an alternating path $(u_1,\ldots,u_l)$ such that
	(a) for all odd $i<l$, we have $(u_i,u_{i+1})\in M(\vecmv[u])$; for all even $i<l$, we have $(u_i,u_{i+1})\in M(\vect{y'})$;
	(b) from $M(\vecmv[u])$ to $M(\vect{y'})$, vertices $\{u_2,u_4,\ldots\}$ get worse, vertices $\{u_3,u_5,\ldots\}$ get better.
	
	Hence the symmetric difference between $M(\vecy)$ and $M(\vecmv[u])$ is the alternating path $(u,u_1,\ldots,u_l)$ (recall that $M(\vecy) = M(\vec{y'})\cup \{ (u,u_1) \}$).
	It is easy to see that statement (a) holds, and statement (b) holds for vertices $\{u_2,\ldots,u_l\}$.
	
	Now consider vertex $u_1$, which is matched to $u$ in $M(\vecy)$, and matched to $u_2$ in $M(\vecmv[u])$.
	
	If $u_1$ is passively matched (by $u$) in $M(\vecy)$, then we know that $u$ has an earlier deadline than $u_1$.
	Hence in $M(\vecmv[u])$, either $u_1$ is active, or passively matched by some $u_2$ with a deadline later than $u$.
	In other words, $u_1$ gets worse from $M(\vecy)$ to $M(\vecmv[u])$.
	
	If $u_1$ matches $u$ actively in $M(\vecy)$, then we know that $u_1$ has an earlier deadline than $u$.
	Hence when $u_1$ is considered in $\vecmv[u]$, the set of unmatched vertices (except for $u$) is identical as in $M(\vecy)$.
	Consequently, $u_1$ actively matches some vertex $u_2$ with $y_{u_2} > y_u$ (otherwise $u_1$ will not match $u$ in $M(\vecy)$).
	In other words, $u_1$ gets worse from $M(\vecy)$ to $M(\vecmv[u])$.	
\end{proofof} 

\begin{proofof}{Lemma~\ref{lemma:dual_fitting}}
	Let $\tilde{\alpha}_u := \expect{\vecy}{\alpha_u}/F$ for all $u\in V$. By the first assumption, 
	\[
	\sum_{u \in V} \tilde{\alpha}_u = \sum_{u \in V} \frac{\expect{\vecy}{{\alpha_u}}}{F} = \frac{1}{F}\expect{\vecy}{\sum_{u \in V} {\alpha_u}} =  \frac{1}{F}\expect{\vecy}{\sum_{(u,v)\in E} x_{uv}}.
	\]
	
	Moreover, $\tilde{\alpha}$ is a feasible dual solution: by the second assumption, $\tilde{\alpha}_u + \tilde{\alpha}_v = \expect{\vecy}{\alpha_u + \alpha_v}/F \ge 1$ for all $(u,v) \in E$.
	By duality, we conclude that
	\[
	\frac{1}{F}\expect{\vecy}{\sum_{(u,v)\in E} x_{uv}} = \sum_{u \in V} \tilde{\alpha}_u \ge \textsf{OPT},
	\]
	where \textsf{OPT} is the optimal primal solution, which is at least the size of a maximum matching.
\end{proofof}
\section{Missing Proofs in Section~\ref{sec:general}} \label{appendix:general}

\subsection{General Version of Lemma~\ref{lemma:general_main_sketch}} \label{subsec:general_main}

In this section we prove the following lemma, which is a general version of Lemma~\ref{lemma:general_main_sketch} (without the simplifying assumptions on $u,v$ we made in Section~\ref{sec:general}).

\begin{lemma} \label{lemma:general_main}
	For any neighbor $u$ of $v$ that has an earlier deadline than $v$, and for any $\vecmv$, we have
	\begin{align*}
		\expect{y_v}{\alpha_u + \alpha_v} \ge f(y_u) \eqdef \min_{\theta} \biggr\{  \min_{\tau \in [\theta,1)} \big\{ \int_{0}^{\theta}g(y_v)dy_v + \min \left\{(1-\theta)(1-g(1^\text{-})-h(1^\text{-})), (\tau - \theta) h(\theta) \right\} \\
		+ (1-\theta) \min \left\{ g(y_u), 1-g(\theta)-h(\theta) \right\} + \theta \min \left\{g(y_u), 1-g(\tau)-h(\tau)\right\} \big\}, \\
		\int_{0}^{\theta}g(y_v)dy_v + (1-\theta)\min \{1-g(1^\text{-})-h(1^\text{-}), h(\theta) \} + (1-\theta) \min \left\{ g(y_u), 1-g(\theta) \right\}  \biggr\}.
	\end{align*}
\end{lemma}

Fix any neighbor $u$ of $v$ with an earlier deadline than $v$, and any $\vecmv$.
Let $\theta$ be the marginal rank of $v$, i.e. $v$ is passive only when $y_v<\theta$.
By Lemma~\ref{lemma:unmatched_neighbor}, we know that when $y_v \geq \theta$, $u$ either is passive or actively matches some vertex with rank at most $\theta$. 

We define in the following two lists of thresholds $\{\theta_i\}_{i=0}^{m+1}$ and $\{\tau_i\}_{i=0}^{m+1}$ that captures the matching statuses of $u$ when $y_v$ is smaller than the marginal rank $\theta$.

Imagine that we decrease $y_v$ continuously starting from $y_v = \theta$. Let $\theta_0 = \theta$ and $\tau_0 = \theta$.
We define $y_v = \theta_{i+1}^\text{-}$ to be the first moment after $\theta_i$ when $u$ actively matches some vertex $z_{i+1}$ with $y_{z_{i+1}} > \tau_i$.
For convenience of description, we say that $u$ actively matches a vertex with rank $1$ if $u$ is unmatched (by definition, the gain of $\alpha_u$ is $0$ in both descriptions since $1-g(1)=0$).
Define $\tau_{i+1} := y_{z_{i+1}}$.
Let $\theta_m$ be the last non-zero threshold.
For convenience, we define $\theta_{m+1} = 0$ and $\tau_{m+1} = 1$.
By definition we have the following fact.

\begin{fact} \label{fact:thresholds}
	There exists a sequence of non-increasing thresholds $\{\theta_i\}_{i=0}^{m+1}$ and a sequence of non-decreasing thresholds $\{\tau_i\}_{i=0}^{m+1}$ such that
	\begin{enumerate}
		\item for all $0\le i\le m$ and $y \in [\theta_{i+1}, \theta_i)$, $u$ is passive or actively matches some vertex with rank at most $\tau_i$ in $M(y_v=y,\vecmv[v])$;
		\item for all $1 \le i\le m$, $u$ actively matches a vertex $z_i$ with rank $\tau_i$ in $M(y_v = \theta_i^{\text{-}} , \vecmv[v])$.
	\end{enumerate}
\end{fact}

For all $i\in[m]$, let $w_i$ be be vertex that actively matches $v$ in $M(y_v=\theta_i^\text{-}, \vecmv)$.

Observe that all $w_i$'s must be different, e.g. the deadline of $w_{i+1}$ must be earlier than $w_i$, in order for $w_{i+1}$ to match $v$ in $M(y_v = \theta_{i+1}^\text{-},\vecmv[v])$.
Moreover, we know that in $M(y_v = \theta,\vecmv)$, each $w_i$ matches a vertex with rank $\theta_i$, since $w_i$ chooses $v$ in $M(y_v = \theta_i^{\text{-}},\vecmv)$ but not in $M(y_v = \theta_i,\vecmv)$.

\begin{lemma} \label{lem:compensate}
	For all $y \in [\theta, \tau_i)$, if $v$ is unmatched in $M(y_v = y, \vecmv[v])$, then $v$ is the victim of $w_i$.
\end{lemma}
\begin{proof}
	Let $\vecy_1 = (y_v = y, \vecmv)$, where $y \in [\theta, \tau_i)$.
	Trivially, $v$ is an unmatched neighbor of $w_i$ in $M(\vecy_1)$. 
	To prove that $v$ is the victim of $w_i$, it suffices to show that $w_i$ is active in $M(\vecy_1)$ and $v$ becomes matched when we remove $w_i$ from the graph.
	
	Let $\vecy_2 = (y_v = \theta_i^{\text{-}}, \vecmv)$.
	By definition, $v$ is passively matched by $w_i$ and $u$ actively matches $z_i$ with $y_{z_i} = \tau_i$ in $M(\vecy_2)$.
	It is easy to see that $w_i$ is also active in $M(\vecy_1)$, as otherwise, $w_i$ should still be passive in $M(\vecy_2)$ given that $v$ does not affect any decisions before $w_i$'s deadline.
	
	Let $\vecy_3$ be the ranks by removing the $w_i$-th entry from $\vecy_1$.
	Assume for contrary that $v$ is unmatched in $M(\vecy_3)$.
	Then we have $M(\vecy_3) = M(\vecy_2) \setminus \{ (w_i,v) \}$.
	This implies that $u$ actively matches $z_i$ in $M(\vecy_3)$ while $v$ (with rank $y_v < y_{z_i} = \tau_i$) is unmatched, which is a contradiction.
\end{proof}

Equipped with Lemma~\ref{lem:compensate}, we first give a lower bound on the expected gain of $\alpha_v$.
For notational convenience, we define a new function $\phi: [0,1]\rightarrow [0,1]$ such that $\phi(y) := 1-g(y)-h(y)$.
Recall by definition of $g$ and $h$, $\phi$ is a non-increasing function with $\phi(1) = 0$.

\begin{lemma} \label{lemma:lower_bound_a_v}
	$\expect{y_v}{\alpha_v} \ge \int_{0}^{\theta}g(y_v)dy_v + \min \left\{(1-\theta)\phi(1^\text{-}), \sum_{i=0}^{m} (\tau_i - \theta) h(\theta_i) \right\}$.
\end{lemma}
\begin{proof}
	By definition, $v$ is passive when $y_v<\theta$. Hence we have $\expect{y_v}{\alpha_v \cdot \idr{y_v < \theta}} = \int_{0}^{\theta}g(y_v)dy_v$, which corresponds to the first term on the RHS.
	
	For all $y_v \geq \theta$, $v$ is either active or unmatched.
	In the first case, let $p$ be matched passively by $v$ in $M(y_v,\vecmv[v])$.
	We know that $v$ gains $1-g(y_p)$ during the gain sharing phase and gives away $h(y_p)$ to its victim (if any),
	which implies $\alpha_v \ge 1-g(y_z)-h(y_z)  = \phi(y_z) \ge \phi(1^{\text{-}})$.
	Hence we have $\expect{y_v}{\alpha_v \cdot \idr{y_v \ge \theta}} \ge (1-\theta)\phi(1^\text{-})$.
	
	In the second case, by Lemma~\ref{lem:compensate}, $v$ is the victim of $w_i$ when $y_v \in [\theta,\tau_i)$.
	Hence $v$ gains $h(\theta_i)$ from $w_i$ in the compensation phase (recall that $w_i$ matches a vertex with rank $\theta_i$ when $y_v \in [\theta, \tau_i)$).
	Putting all compensation (from $w_1,\ldots,w_m$) together, we get $\expect{y_v}{\alpha_v \cdot \idr{y_v \ge \theta}} \ge \sum_{i=0}^{m} (\tau_i - \theta) h(\theta_i)$, which corresponds to the second the term on the RHS.
\end{proof}

\begin{lemma} \label{lemma:lower_bound_a_u_above}
	$\expect{y_v}{\alpha_u \cdot \idr{y_v < \theta}} \ge \sum_{i=0}^{m} (\theta_i - \theta_{i+1}) \min \left\{g(y_u),\phi(\tau_i)\right\}$.
\end{lemma}
\begin{proof}
	We partition the interval $[0,\theta)$ into $m+1$ segments: $[\theta_{i+1},\theta_{i})$, for $0\le i \le m$.
	Fix any $i$, and consider $y_v \in [\theta_{i+1},\theta_{i})$.
	If $u$ is passive in $M(y_v,\vecmv[v])$, then we have $\alpha_u \geq g(y_u)$.
	Otherwise, we know that $u$ actively matches a vertex $z$ with $y_z \leq \tau_i$ (by Fact~\ref{fact:thresholds}).
	Hence $u$ gains $1-g(y_z)$ during the gain sharing phase and gives away $h(y_z)$ to its victim (if any), i.e., we have $\alpha_u \ge \phi(y_z) \ge \phi(\tau_i)$.
	Summing up the gain from the $m+1$ segments concludes the proof.
\end{proof}

Observe that for lower bounding $\expect{y_v}{\alpha_u + \alpha_v}$,  we shall consider the total gain of $\alpha_u + \alpha_v$. We may omit the compensation from $u$ to $v$, since it does not change the summation. For analysis convenience, we assume $v$ is never a victim of $u$.

\begin{lemma} \label{lemma:lower_bound_a_u_below}
	$\expect{y_v}{\alpha_u \cdot \idr{y_v \ge \theta}} \ge (1-\theta) \min\{g(y_u), \phi(\theta)\}$. Moreover, if $\tau_m=1$, then we have
	$\expect{y_v}{\alpha_u \cdot \idr{y_v \ge \theta}} \ge (1-\theta) \min\{g(y_u), 1-g(\theta)\}.$
\end{lemma}
\begin{proof}
	For all $y_v \ge \theta$, $u$ is either passive or actively matches a vertex with rank at most $\theta$. Therefore, $\alpha_u \ge \min \{g(y_u), \phi(\theta)\}$. Integrating $y_v$ from $\theta$ to $1$ gives the first statement.
	
	When $\tau_m=1$, we know that $u$ is unmatched when $y_v = \theta_m^{\text{-}}$.
	Fix any $\vecy = (y_v,\vecmv[v])$, where $y_v \ge \theta$.
	We show that $u$ does not have any unmatched neighbor other than $v$ in $M(\vecy)$, which implies that $u$ does not have a victim and hence $\alpha_u \ge \min \{g(y_u), 1-g(\theta)\}$.
	
	Suppose otherwise, let $z \neq v$ be the unmatched neighbor of $u$ in $M(\vecy)$.
	
	Let $\vecy_1 = (y_v = \theta_m^{\text{-}}, \vecmv)$.
	We know that $u$ is matched in $M(\vecy)$ and unmatched in $M(\vecy_1)$.
	Consider the partial matchings produced right after $u$'s deadline when \ranking is run with $\vecy$ and $\vecy_1$, respectively.
	We denote the matchings by $M_u(\vecy)$ and $M_u(\vecy_1)$, respectively.
	It is easy to see that the symmetric difference between $M_u(\vecy)$ and $M_u(\vecy_1)$ is an alternating path, with $u$ being one endpoint.
	Observe that $v$ is matched in $M_u(\vecy_1)$ (as $u$ is unmatched), and is unmatched in $M_u(\vecy)$ (as it is not passive in $M(\vecy)$).
	Hence $v$ is the other end point of the alternating path.
	Consequently, we know that $z$ is unmatched in $M_u(\vecy_1)$ (as it is unmatched in $M_u(\vecy)$), which is a contradiction as its neighbor $u$ is also unmatched in $M_u(\vecy_1)$.
\end{proof}

The next technical lemma shows that the worst case is achieved when there exists only one threshold $\theta_m = \theta$, i.e. $u$ matches some vertex with rank $\tau > \theta$ in $M(y_v=\theta^{\text{-}}, \vecmv)$, and matches a vertex with rank at most $\tau$ for all $y_v < \theta$.

\begin{lemma} \label{lemma:technical_lower_bound}
	Given that $h(y)/y$ is a non-increasing function, we have
	\begin{align*}
		& \min \bigg\{(1-\theta)\phi(1^{\text{-}}), \sum_{i=0}^{m} (\tau_i - \theta) h(\theta_i) \bigg\} + \sum_{i=0}^{m} (\theta_i-\theta_{i+1}) \min \{g(y_u),\phi(\tau_i)\} \\
		\ge & \min_{i} \bigg\{ \min \left\{(1-\theta)\phi(1^{\text{-}}), (\tau_i - \theta) h(\theta) \right\} + \theta \min \left\{g(y_u),\phi(\tau_i)\right\} \bigg\}.
	\end{align*}
\end{lemma}
\begin{proof}
	Consider the first term of LHS.
	If $(1-\theta)\phi(1^{\text{-}}) < \sum_{i=0}^{m} (\tau_i - \theta) h(\theta_i)$, we have
	\begin{align*}
		\text{LHS} \ge & (1-\theta)\phi(1^{\text{-}}) +  \sum_{i=0}^{m} (\theta_i - \theta_{i+1}) \min \left\{g(y_u),\phi(\tau_m)\right\} \\
		= & (1-\theta)\phi(1^{\text{-}}) +  \theta \min \left\{g(y_u),\phi(\tau_m)\right\} \ge \text{RHS}.
	\end{align*}
	
	If $(1-\theta)\phi(1^{\text{-}}) \ge \sum_{i=0}^{m} (\tau_i - \theta) h(\theta_i)$, we have
	\begin{align*}
		\text{LHS} = & \sum_{i=0}^{m} (\tau_i - \theta) h(\theta_i) + \sum_{i=0}^{m} (\theta_i - \theta_{i+1}) \min \{g(y_u),\phi(\tau_i)\} \\
		= & \sum_{i=0}^{m} (\theta_i - \theta_{i+1}) \left( \frac{\tau_i - \theta}{\theta_{i}-\theta_{i+1}} h(\theta_i) + \min \{g(y_u),\phi(\tau_i)\} \right) \\
		\ge & \sum_{i=0}^{m} (\theta_i - \theta_{i+1}) \left( \frac{\tau_i - \theta}{\theta} h(\theta) + \min \{g(y_u),\phi(\tau_i)\} \right) \\
		\ge & \sum_{i=0}^{m} (\theta_i - \theta_{i+1}) \min_{j} \left\{ \frac{\tau_j - \theta}{\theta} h(\theta) + \min \{g(y_u),\phi(\tau_j)\}  \right\} \\
		= & \min_{j} \bigg\{ (\tau_j - \theta) h(\theta)  + \theta \min \{g(y_u),\phi(\tau_i)\} \bigg\} \ge \text{RHS},
	\end{align*}
	where the first inequality follows from $\frac{\tau_i - \theta}{\theta_{i}-\theta_{i+1}} h(\theta_i) \ge (\tau_i - \theta)\frac{h(\theta_i)}{\theta_{i}} \ge (\tau_i - \theta)\frac{h(\theta)}{\theta}$.
\end{proof}

\begin{proofof}{Lemma~\ref{lemma:general_main}}
	By Lemma~\ref{lemma:lower_bound_a_v} and \ref{lemma:lower_bound_a_u_above}, we have
	\begin{align*}
		& \expect{y_v}{\alpha_v} + \expect{y_v}{\alpha_u \cdot \idr{y_v < \theta}}\\
		\ge & \int_{0}^{\theta}g(y_v)dy_v + \min \left\{(1-\theta)\phi(1^{\text{-}}), \sum_{i=0}^{m} (\tau_i - \theta) h(\theta_i) \right\}
		+ \sum_{i=0}^{m} (\theta_i - \theta_{i+1}) \min \{g(y_u),\phi(\tau_i)\} \\
		\ge & \int_{0}^{\theta}g(y_v)dy_v + \min_{i} \bigg\{ \min \{(1-\theta)\phi(1^{\text{-}}), (\tau_i - \theta) h(\theta) \}
		+ \theta \min \{g(y_u),\phi(\tau_i)\} \bigg\},
	\end{align*}
	where the last inequality comes from Lemma~\ref{lemma:technical_lower_bound}.
	
	Combining with Lemma~\ref{lemma:lower_bound_a_u_below} (which gives different lower bounds for $\expect{y_v}{\alpha_u \cdot \idr{y_v \ge \theta}}$ depending on whether $\tau_m = 1$), we prove Lemma~\ref{lemma:general_main} for two cases, depending on whether $\tau_m = 1$.
	
	If $\tau_m < 1$, we have (recall that $\phi(\theta) = 1 - g(\theta) - h(\theta) \leq 1-g(\theta)$)
	\begin{align*}
		&\expect{y_v}{\alpha_u + \alpha_v} = \expect{y_v}{\alpha_v} + \expect{y_v}{\alpha_u \cdot \idr{y_v < \theta}} + \expect{y_v}{\alpha_u \cdot \idr{y_v \ge \theta}} \\
		\ge & \int_{0}^{\theta}g(y_v)dy_v + \min_{\tau \in [\theta,1)} \big\{ \min \{(1-\theta)\phi(1^{\text{-}}), (\tau - \theta) h(\theta) \}
		+  \theta \min \{g(y_u),\phi(\tau)\} \big\} \\
		& + (1-\theta) \min\{g(y_u), \phi(\theta)\},
	\end{align*}
	which corresponds to the first term of the outer most $\min$ in the expression of Lemma~\ref{lemma:general_main}.
	
	If $\tau_m = 1$, we have 
	\begin{align*}
		\expect{y_v}{\alpha_u + \alpha_v}
		\ge & \int_{0}^{\theta}g(y_v)dy_v + \min_{\tau \in [\theta,1]} \big\{ \min \{(1-\theta)\phi(1^{\text{-}}), (\tau - \theta) h(\theta) \}
		+ \theta \min \{g(y_u),\phi(\tau)\} \big\} \\
		& + (1-\theta) \min\{g(y_u), 1-g(\theta)\} \\
		\ge & \min\Bigg\{ \min_{\tau \in [\theta,1)} \bigg\{ \int_{0}^{\theta}g(y_v)dy_v + \min \{(1-\theta)\phi(1^{\text{-}}), (\tau - \theta) h(\theta) \}
		+ \theta \min \{g(y_u),\phi(\tau)\}  \\
		& \qquad\qquad\qquad + (1-\theta) \min\{g(y_u), \phi(\theta)\} \bigg\} ,\\
		& \qquad\quad \int_{0}^{\theta}g(y_v)dy_v + (1-\theta)\min \{\phi(1^{\text{-}}), h(\theta)\}
		+ (1-\theta) \min\{g(y_u), 1-g(\theta)\} \Bigg\}.
	\end{align*}
	
	Taking the minimum over all possible $\theta$'s concludes the proof.
\end{proofof}

\subsection{Lower Bound of the Competitive Ratio} \label{subsec:ratio}

Recall that
\begin{align*}
	f(y_u) \eqdef \min_{\theta} \biggr\{ \min_{\tau \in [\theta, 1)} \big\{ \int_{0}^{\theta}g(y_v)dy_v + \min \left\{(1-\theta)(1-g(1^\text{-})-h(1^\text{-})), (\tau - \theta) h(\theta) \right\} \\
	+ (1-\theta) \min \left\{ g(y_u), 1-g(\theta)-h(\theta) \right\} + \theta \min \left\{g(y_u), 1-g(\tau)-h(\tau)\right\} \big\}, \\
	\int_{0}^{\theta}g(y_v)dy_v + (1-\theta)\min \{1-g(1^\text{-})-h(1^\text{-}), h(\theta) \} + (1-\theta) \min \left\{ g(y_u), 1-g(\theta) \right\} \biggr\}.
\end{align*}

We construct the functions $g,h$ explicitly as follows (refer to Figure~\ref{fig:gxhx}).
\begin{align*}
	g(x) = \begin{cases}
		k_g^1x+b, &x \in [0, t]\\
		k_g^2(x-t) + k_g^1 t+b, &x \in (t,1)\\
		1 & x = 1
	\end{cases}
	,\quad
	h(x) = \begin{cases}
		k_h^1x, &x \in [0,t]\\
		k_h^2(x-t) + k_h^1 t, & x \in (t,1)\\
		0 & x = 1
	\end{cases},
\end{align*}
where $t = 0.3,k_g^1 = 0.21,k_g^2=0.1, b=0.46, k_h^1=0.26,k_h^2=0.17$.
It is easy to see that both $g,h$ are non-decreasing in $[0,1)$ and $h(x)/x$ is non-increasing.

\begin{figure}[H]
	\centering
	\centering\includegraphics[width = 0.6\textwidth]{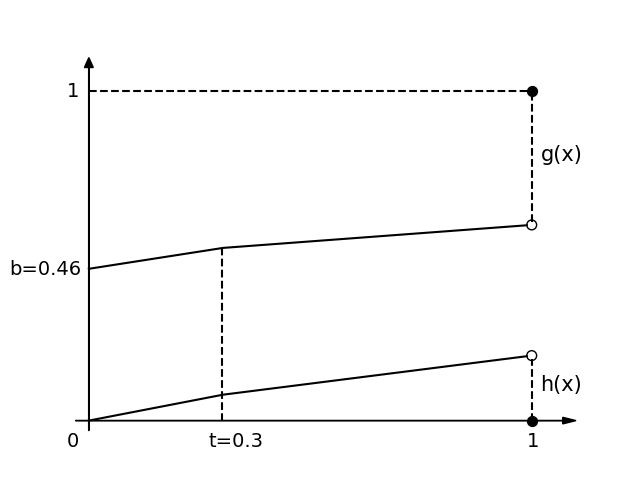}
	\caption{$g(x) \text{ and } h(x)$}
	\label{fig:gxhx}
\end{figure}

We first simplify the expression of $f(y_u)$.
Recall that we define $\phi(y) = 1-g(y)-h(y)$, which is a decreasing function with $\phi(1) = 0$.
By definition of $g,h$ stated above, we have $\phi(1^\text{-}) = 0.21$.

Observe that $(1-\theta)\phi(1^\text{-})  > (1-\theta)\cdot 0.197 = (1-\theta)h(1^\text{-}) \geq (\tau-\theta)h(\theta)$ for all $\tau, \theta$. Hence,
\[
f(y_u) = \min\left\{\min_{\theta \le \tau < 1}\psi_1(y_u,\theta,\tau),\min_{\theta\leq 1}\psi_2(y_u,\theta)\right\},
\]
where
\begin{align*}
	\psi_1 (y_u,\theta,\tau) \eqdef& \int_{0}^{\theta}g(y_v)dy_v + (\tau - \theta) h(\theta) + (1-\theta) \min \left\{ g(y_u), \phi(\theta) \right\} 
	+ \theta \min \left\{g(y_u), \phi(\tau)\right\}, \\
	\psi_2 (y_u,\theta) \eqdef& \int_{0}^{\theta}g(y_v)dy_v + (1 - \theta) h(\theta)  + (1-\theta) \min \left\{ g(y_u), 1-g(\theta) \right\}.
\end{align*}

The following lemma implies that \ranking is $0.5211$-competitive.

\begin{lemma}
	$\int_0^1 f(y_u)dy_u > 0.5211$.
\end{lemma}
\begin{proof}
	As we will show in Lemma~\ref{lemma:psi2} and Lemma~\ref{lemma:psi1} below, we have
	\[
	f(y_u) = \min\left\{\min_{\theta \le \tau < 1}\psi_1(y_u,\theta,\tau),\min_{\theta\leq 1}\psi_2(y_u,\theta)\right\} \ge \min\{g(y_u),0.5349\}.
	\]
	
	Let $y_u^*$ be such that $g(y_u^*)=0.5349$, we have
	\begin{align*}
		\int_{0}^{1} f(y_u) dy_u = & \int_{0}^{y_u^*} f(y_u) dy_u + \int_{y_u^*}^{1} f(y_u) dy_u
		=  \int_{0}^{y_u^*} g(y_u) dy_u + (1-y_u^*) \cdot 0.5349 > 0.5211.
	\end{align*}
\end{proof}

We first consider the easier one, $\psi_2$.

\begin{lemma}\label{lemma:psi2}
	For any $\theta\leq 1$, we have $\psi_2 (y_u,\theta)>\min\{g(y_u),0.5349\}$.
\end{lemma}
\begin{proof}
	First, if $\theta = 1$, we have $\psi_2(y_u,\theta) = \int_0^1g(y_v)dy_v \approx 0.5381 > 0.5349$. Now consider $\theta < 1$.
	
	If $g(y_u) < 1-g(\theta)$, we have $\psi_2 (y_u,\theta) = \int_{0}^{\theta}g(y_v)dy_v + (1 - \theta) h(\theta)  + (1-\theta)g(y_u)$. Thus
	\begin{align*}
		\frac{\partial \psi_2}{\partial \theta} &= g(\theta) + (1-\theta)h'(\theta) - h(\theta) - g(y_u),\\
		\frac{\partial^2 \psi_2}{\partial \theta^2} &= g'(\theta) - 2 \cdot h'(\theta) = \begin{cases}
			k_g^1-2k_h^2 =-0.31 &\theta < t\\
			k_g^2-2k_h^2 =-0.24 &\theta > t
		\end{cases}.
	\end{align*}
	
	Hence the minimum of $\psi_2 (y_u,\theta)$ is achieved at $\arg\min_{\theta<1} \{ \psi_2 (y_u,\theta) \} \in \{0,t,1^\text{-}\}$.
	Note that
	\begin{align*}
		\psi_2 (y_u,0) &= h(0) + g(y_u) \ge g(y_u), \\
		\psi_2 (y_u,t) &= \int_{0}^{t}g(y_v)dy_v + (1 - t) h(t)  + (1-t)g(y_u) >  t\cdot0.673 + (1-t) \cdot g(y_u) > g(y_u),\\
		\psi_2(y_u,1^\text{-}) &= \int_{0}^{1}g(y_v)dy_v \approx 0.5381 > 0.5349.
	\end{align*}
	We have $\min_{\theta<1}\{ \psi_2 (y_u,\theta) \} \geq \min\{g(y_u),0.5349\}$, as claimed.
	
	If $g(y_u) \geq 1-g(\theta)$, we have $\psi_2 (y_u,\theta) = \int_{0}^{\theta}g(y_v)dy_v + (1 - \theta) h(\theta)  + (1-\theta)(1-g(\theta))$. Thus
	\begin{align*}
		&\frac{\partial \psi_2}{\partial \theta} = g(\theta) + (1-\theta)h'(\theta) - h(\theta) - (1-g(\theta)) - (1-\theta)g'(\theta),\\
		&\frac{\partial^2 \psi_2}{\partial \theta^2} = 3\cdot g'(\theta) - 2\cdot h'(\theta)=\begin{cases}
			0.11  &\theta < t\\
			-0.04 &\theta > t
		\end{cases}.
	\end{align*}
	
	We calculate the zero point of $\frac{\partial \psi_2}{\partial \theta}$ in $[0,t)$, i.e., let $\frac{\partial \psi_2}{\partial \theta} = 0$, we have solution
	\begin{align*}
		&\theta^* = \frac{1+k_g^1-k_h^1-2b}{3k_g^1-2k_h^1}\approx 0.273.
	\end{align*}
	
	\vspace*{-20pt}
	\begin{figure}[H]
		\centering
		\subfigure[$\psi_2(y_u,\theta)$]{\centering\includegraphics[width = 0.45\textwidth]{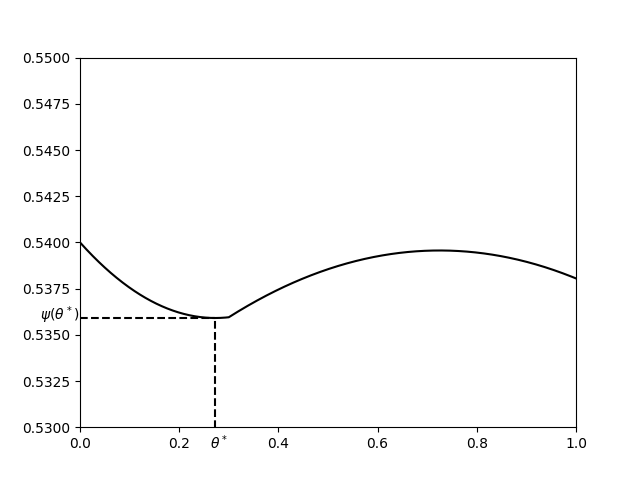}}
		\subfigure[$\frac{\partial \psi_2}{\partial \theta}$]{\centering\includegraphics[width = 0.45\textwidth]{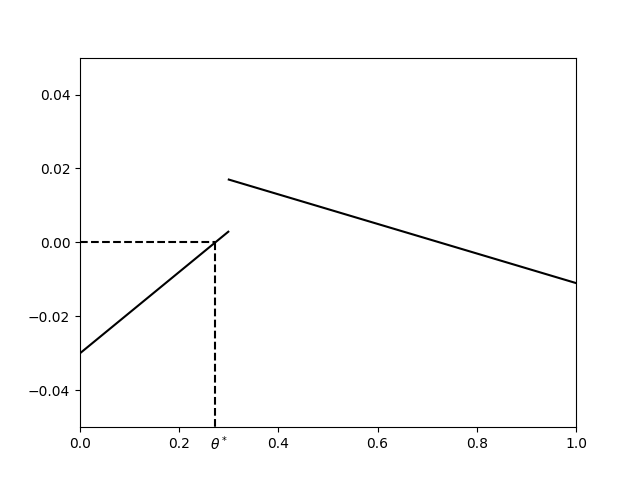}}
		\caption{$\psi_2(y_u,\theta)$ and $\frac{\partial \psi_2}{\partial \theta}$}
		\label{fig:tau=1}
	\end{figure}
	Thus, for any fixed $y_u$, $\psi_2(y_u,\theta)$ is decreasing in $[0,\theta^*]$ and increasing in $(\theta^*,t)$. So the minimum of $\psi_2(y_u,\theta)$ in $[0,t]$ is achieved at $\theta^*$. Also, since $\frac{\partial^2 \psi_2}{\partial \theta^2}<0$ in $(t,1)$, the minimum of $\psi_2(y_u,\theta)$ in $[t,1)$ is achieved at either $t$ or $1^\text{-}$. To sum up, the overall minimum is achieved at either $\theta^*$ or $1^\text{-}$:
	\begin{align*}
		\psi_2 (y_u,\theta^*) &= \int_{0}^{\theta^*}g(y_v)dy_v + (1 - \theta^*) h(\theta^*)  + (1-\theta^*)(1-g(\theta^*)) \approx 0.5359 >0.5349\\
		\psi_2(y_u,1^\text{-}) &= \int_{0}^{1}g(y_v)dy_v \approx 0.5381 >0.5349,
	\end{align*}
	Thus we have $\min_{\theta<1} \psi_2 (y_u,\theta) = \min\{\psi_2 (y_u,\theta^*),\psi_2 (y_u,1^\text{-})\} > 0.5349$, as required.
\end{proof}

Next we consider $\psi_1$.

\begin{lemma}\label{lemma:psi1}
	For all $\theta \leq \tau < 1$, we have $\psi_1 (y_u,\theta,\tau)>\min \{g(y_u),0.5349\}$.
\end{lemma}
\begin{proof}
	If $g(y_u) \leq \phi(\tau)$, then we have
	$\psi_1(y_u,\theta,\tau) = \int_0^\theta g(y_v)dy_v + (\tau-\theta)h(\theta) + (1-\theta)g(y_u) + \theta g(y_u) \geq g(y_u)$, as required.
	Now consider $g(y_u) > \phi(\tau)$. 
	Observe that
	\begin{align*}
		\frac{\partial \psi_1}{\partial \tau} = h(\theta) - \theta(g'(\tau)+h'(\tau)) \leq 0,
	\end{align*}
	where the last inequality holds since $h(\theta) \leq k_h^1 \theta \le (k_g^2 + k_h^2) \theta \le (g'(\tau)+h'(\tau))\theta$ for all $\tau$.
	Thus, for all $\tau$ we have $\psi_1 (y_u,\theta,\tau) \geq \psi_1(y_u,\theta,1^\text{-})$, i.e., the minimum is achieved when $\tau = 1^\text{-}$.
	
	Depending on whether $g(y_u) \geq \phi(\theta)$, we consider two cases.
	If $g(y_u) < \phi(\theta)$, we have
	\begin{equation*}
		\psi_1(y_u,\theta,\tau) = \int_0^\theta g(y_v)dy_v + (\tau-\theta)\cdot h(\theta) + (1-\theta)\cdot g(y_u) + \theta\cdot \phi(\tau).
	\end{equation*}
	
	For any fixed $\theta$, the minimum is achieved when $\tau=1^\text{-}$, which is
	\begin{align*}
		\psi_1(y_u,\theta,1^\text{-})
		&= \int_0^\theta g(y_v)dy_v + (1^\text{-}-\theta)h(\theta) + (1-\theta)g(y_u) + \theta\cdot \phi(1^\text{-})\\
		&\geq \int_0^\theta g(y_v)dy_v + (1-\theta)h(\theta) + (1-\theta)g(y_u) \geq \psi_2(y_u,\theta) \geq \min\{g(y_u),0.5349\},
	\end{align*}
	where the last inequality follows from Lemma~\ref{lemma:psi2}.
	If $g(y_u) \geq \phi(\theta)$, we have
	\begin{equation*}
		\psi_1(y_u,\theta,\tau) = \int_0^\theta g(y_v)dy_v + (\tau-\theta)\cdot h(\theta) + (1-\theta)\cdot \phi(\theta) + \theta\cdot \phi(\tau),
	\end{equation*}
	the minimum of which is achieved when $\tau=1^\text{-}$. Define $\psi(\theta)$ to be the minimum:
	\begin{align*}
		\psi(\theta) \eqdef
		\psi_1(y_u,\theta,1^\text{-}) &= \int_0^\theta g(y_v)dy_v + (1^\text{-}-\theta)\cdot h(\theta) + (1-\theta)\cdot \phi(\theta) + \theta\cdot \phi(1^\text{-})\\
		&= \int_0^\theta g(y_v)dy_v + (1-\theta)\cdot(1-g(\theta)) + \theta\cdot (1-g(1^\text{-})-h(1^\text{-})).
	\end{align*}
	
	By the following, we have $\frac{\partial \psi}{\partial \theta}>0$ for all $\theta \in (t,1)$, i.e., $\psi(\theta)$ is increasing when $\theta \in (t,1)$.
	\begin{align*}
		&\frac{\partial \psi}{\partial \theta} = g(\theta) + (1-\theta)(-g'(\theta)) - (1-g(\theta)) + 1-g(1^\text{-})-h(1^\text{-}). \\
		&\frac{\partial^2 \psi}{\partial \theta^2} = 3g'(\theta) > 0, \quad \forall \theta \in [0,t) \cup (t,1) \\
		&\left. \frac{\partial \psi}{\partial \theta}\right|_{\theta=0} = 2g(0) - k_g^1 - g(1^\text{-}) - h(1^\text{-}) \approx -0.08 <0\\
		&\left. \frac{\partial \psi}{\partial \theta}\right|_{\theta=t} = 2g(t) - (1-t)k_g^2 - g(1^\text{-}) - h(1^\text{-}) = 0.186 >0.
	\end{align*}
	\vspace*{-20pt}
	\begin{figure}[H]
		\centering
		\subfigure[$\psi(\theta)$]{\centering\includegraphics[width = 0.45\textwidth]{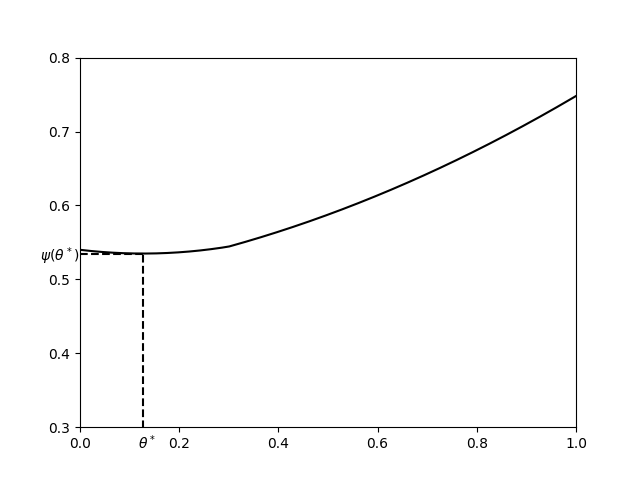}}
		\subfigure[$\frac{\partial \psi}{\partial \theta}$]{\centering\includegraphics[width = 0.45\textwidth]{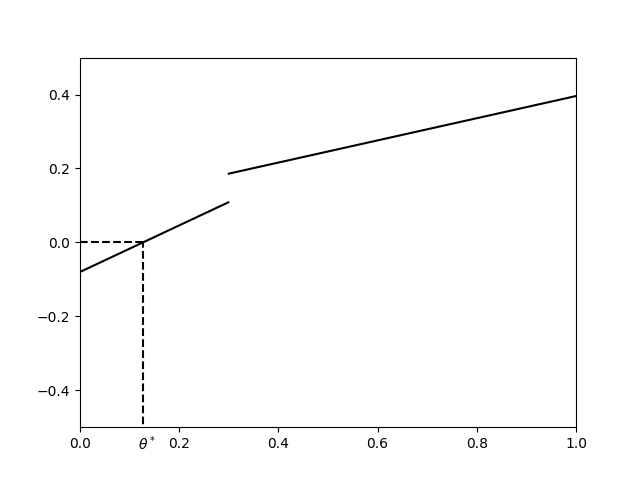}}
		\caption{$\psi(\theta)$ and $\frac{\partial \psi}{\partial \theta}$}
		\label{fig:tau=1m}
	\end{figure}
	
	Hence the minimum of $\psi(\theta)$ is achieved when $\theta \in (0,t]$.
	Let $\frac{\partial \psi}{\partial \theta} = 0$, we have
	\begin{align*}
		& \frac{\partial \psi}{\partial \theta} = g(\theta) - (1-\theta)g'(\theta) - (1-g(\theta)) + 1-g(1^\text{-})-h(1^\text{-}) = 0 \\
		\Longleftrightarrow \quad & 3k_g^1\theta + 2b-g(1^\text{-})-h(1^\text{-})-k_g^1 =0.\\
		\Longleftrightarrow \quad &\theta^* = \frac{k_g^1+g(1^\text{-})+h(1^\text{-})-2b}{3k_g^1} \approx 0.127.
	\end{align*}
	
	Thus for all $\theta \leq \tau < 1$, we have $\psi_1(y_u,\theta,\tau) \geq \psi(\theta) \geq \psi(\theta^*) \approx 0.5349$, as claimed.
\end{proof}

\section{Hardness Results}

\begin{proofof}{Theorem~\ref{th:problem_hardness}}
	Consider the following hard instance. Let $k$, $h$ be integer parameters, and $n := \sum_{i=0}^h k^i = \frac{k^{h+1}-1}{k-1}$ be the number of vertices on each side of a bipartite graph.
	In the following, we construct a bipartite graph on vertices $U\cup V$, where $U = \{ u_1,\ldots,u_n \}$ and $V = \{v_1,\ldots,v_{n-k^h}, b_1,\ldots,b_{k^h}\}$.
	It is easy to check by our construction that the graph is bipartite, but $U, V$ does not correspond to the two sides of the bipartite graph.
	
	\paragraph{Hard Instance.}
	Refer to Figure~\ref{fig:phard} (an illustrating example with $k=3$ and $h=2$).
	At the beginning, vertex $u_1$ arrives, together with all its $k+1$ neighbors (children). Let the deadline of $u_1$ be reached immediately.
	Then we choose uniformly at random $k$ vertices from the $k+1$ neighbors of $u_1$ to be $u_2,\ldots,u_{k+1}$.
	Let the remaining vertex be $v_1$.
	We repeat the procedure for vertices $u_2,\ldots,u_{k+1}$, i.e., each vertex $u_i$ has $k+1$ children, among which $k$ vertices are chosen to be $u_{(i-1)k+2}\ldots,u_{ik+1}$ while the remaining one becomes $v_i$, and let the deadline of $u_i$ be reached immediately.
	We continue building the tree for $h$ levels.
	Note that in level $i$,  there are $k^i$ vertices (excluding the $v$ vertices).
	Hence the tree with $h$ levels has $n$ vertices.
	
	\begin{figure}[htbp]
		\centering
		\includegraphics[width=0.6\textwidth]{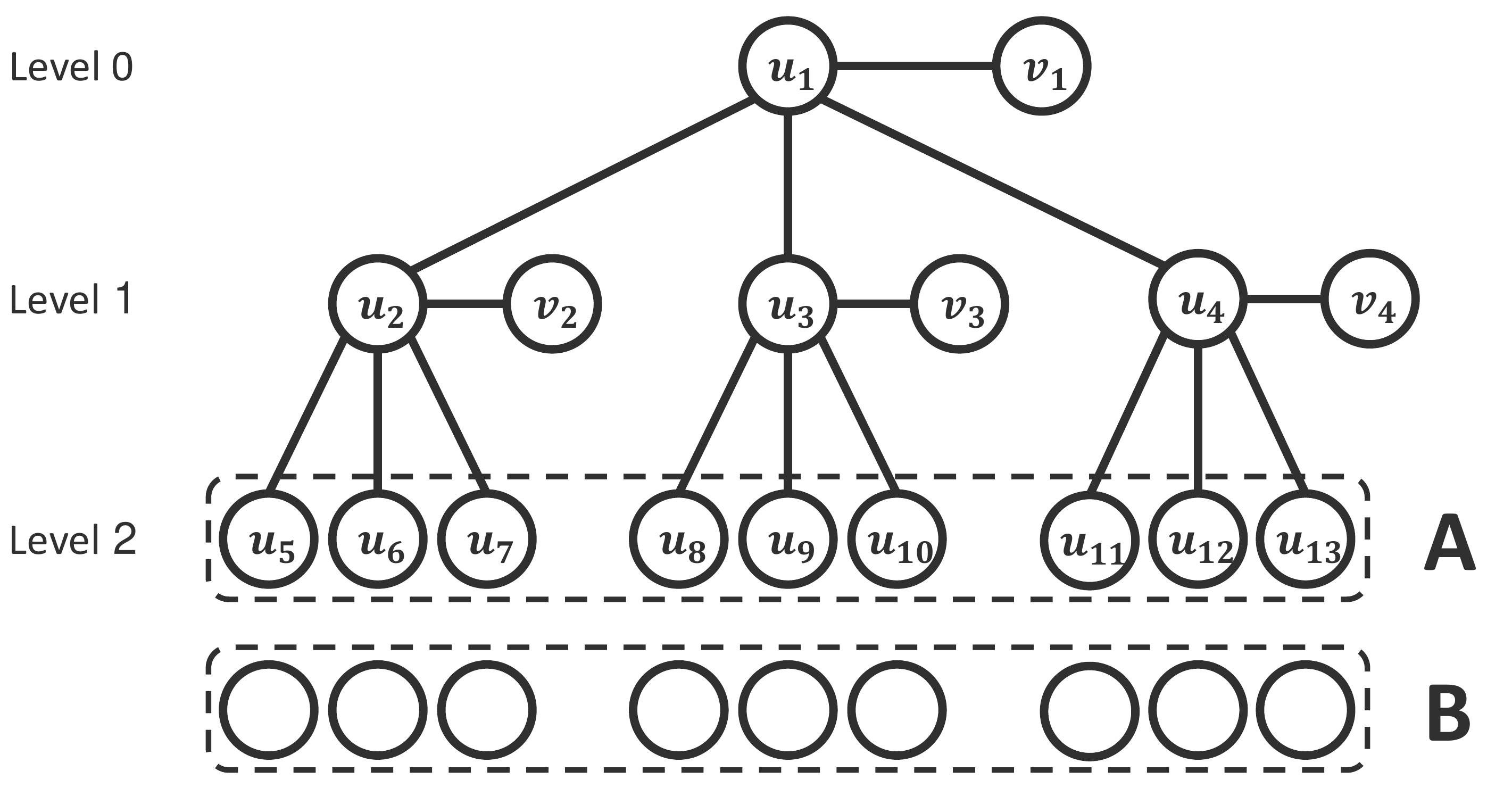}
		\caption{Hard instance for any online algorithm: an illustrating example with $k=3$ and $h=2$.}\label{fig:phard}
	\end{figure}
	
	At last, we pick a random permutation $A = (a_1,a_2,\ldots,a_{k^h})$ of the $k^h$ vertices $\{u_{n-k^h+1},\ldots, u_n\}$ at level $h$.
	Let vertices $B = \{b_1,\ldots,b_{k^h}\}$ arrive ($b_1$ arrives first and $b_{k^h}$ last), such that $b_i$ is connected to vertices $a_i,\ldots, a_n$, and the deadline of $b_i$ is reached immediately when it arrives.
	
	Let the deadlines of vertices in $A\cup (V\setminus B)$ be reached at the end.
	
	\paragraph{Competitive Ratio.}
	First observe that graph $G$ has a perfect matching, by matching $u_i$ to $v_i$ (for all $i\in[n-k^h]$) and $a_i$ to $b_i$ (for all $i\in[k^h]$).
	Now we consider any online algorithm.
	Note that when the deadline of $u_i$ is reached, it is not worse to match $u_i$ if it has an unmatched neighbor: if we do not match $u_i$, then the symmetric difference is an alternating path, thus the number of vertices matched does not increase.
	Hence we assume w.l.o.g. that all vertices in $U\setminus A$ will be matched eventually.
	
	Let $p_i$ be the probability that a vertex $u_x$ from level $i$ is matched when the deadline of its parent $u_y$ in the tree is reached.
	Note that $p_i$ is also the probability that $v_y$ is matched, as $v_y$ is chosen uniformly at random among the $k+1$ children of $u_y$.
	Observe that we have $p_i = \frac{1-p_{i-1}}{k+1}$, where $p_0 = 1$.
	It is easy to check (by induction) that for all $i \leq h$, we have
	\begin{equation*}
		p_i = \frac{1}{k+2}\left(1 - \left(\frac{-1}{k+1}\right)^i \right)
	\end{equation*}
	
	Hence before vertex $b_1$ arrives, each vertex from $A$ is matched with probability $p_h$.
	By the standard water-filling algorithm, it is easy to see that the expected number of matched vertices from $B$ (at the end of the algorithm) is $t$ such that
	\begin{equation*}
		\frac{1}{k^h} + \frac{1}{k^h-1} + \ldots + \frac{1}{k^h - t + 1} = 1 - p_h = \frac{k+1}{k+2} + \frac{1}{k+2} \left(\frac{-1}{k+1} \right)^h .
	\end{equation*}
	
	When $h$ tends to infinity, we have $t \approx (1- e^{-\frac{k+1}{k+2}})\cdot k^h$.
	Hence the competitive ratio is
	\begin{equation*}
		\frac{2t + |U\setminus A| + \frac{1}{k+2}|A| + \frac{1}{k+2}|V\setminus B|}{2n} = \frac{2(1-e^{-\frac{k+1}{k+2}})\cdot k^h\cdot (k-1) + k^h-1}{2(k^{h+1}-1)} + \frac{1}{2(k+2)},
	\end{equation*}
	which tends to $\frac{k-1}{k}(1-e^{-\frac{k+1}{k+2}}) + \frac{1}{2k} + \frac{1}{2(k+2)}$ when $h$ tends to infinity.
	For $k = 7$, the ratio is $\frac{62}{63} - \frac{6}{7}\cdot e^{-\frac{8}{9}}\approx 0.631745$.
\end{proofof}

\begin{proofof}{Theorem~\ref{th:ranking_hardness}}
	Consider the following hard instance. Let $k$, $h$ be integer parameters, and $n := k\cdot h$ be the number of vertices on each side of a bipartite graph $G$.
	In the following, we construct a bipartite graph on vertices $U\cup V$, where $U = \{ u_1,\ldots,u_n \}$ and $V = \{v_1,\ldots,v_{n}\}$.
	As before, it is easy to check by our construction that the graph is bipartite, but $U, V$ does not correspond to the two sides of the bipartite graph.
	
	\paragraph{Hard Instance.}
	Refer to Figure~\ref{fig:rhard}.
	For all $i\in [n]$, let $u_i$ be the only neighbor of $v_i$. 
	We group every $k$ consecutive vertices in $U$ as a group, i.e., let $U = \cup_{i\in [h]} U_i$, where the $i$-th group $U_i = \{ u_{(i-1)k+1},u_{(i-1)k+2},\ldots,u_{ik} \}$.
	Let there be an edge between $u_i$ and $u_j$ if they are from two consecutive groups, respectively.
	In other words, we form a complete bipartite graph between any two consecutive groups $U_i$ and $U_{i+1}$.
	\begin{figure}[htb]
		\centering
		\includegraphics[width=0.6\textwidth]{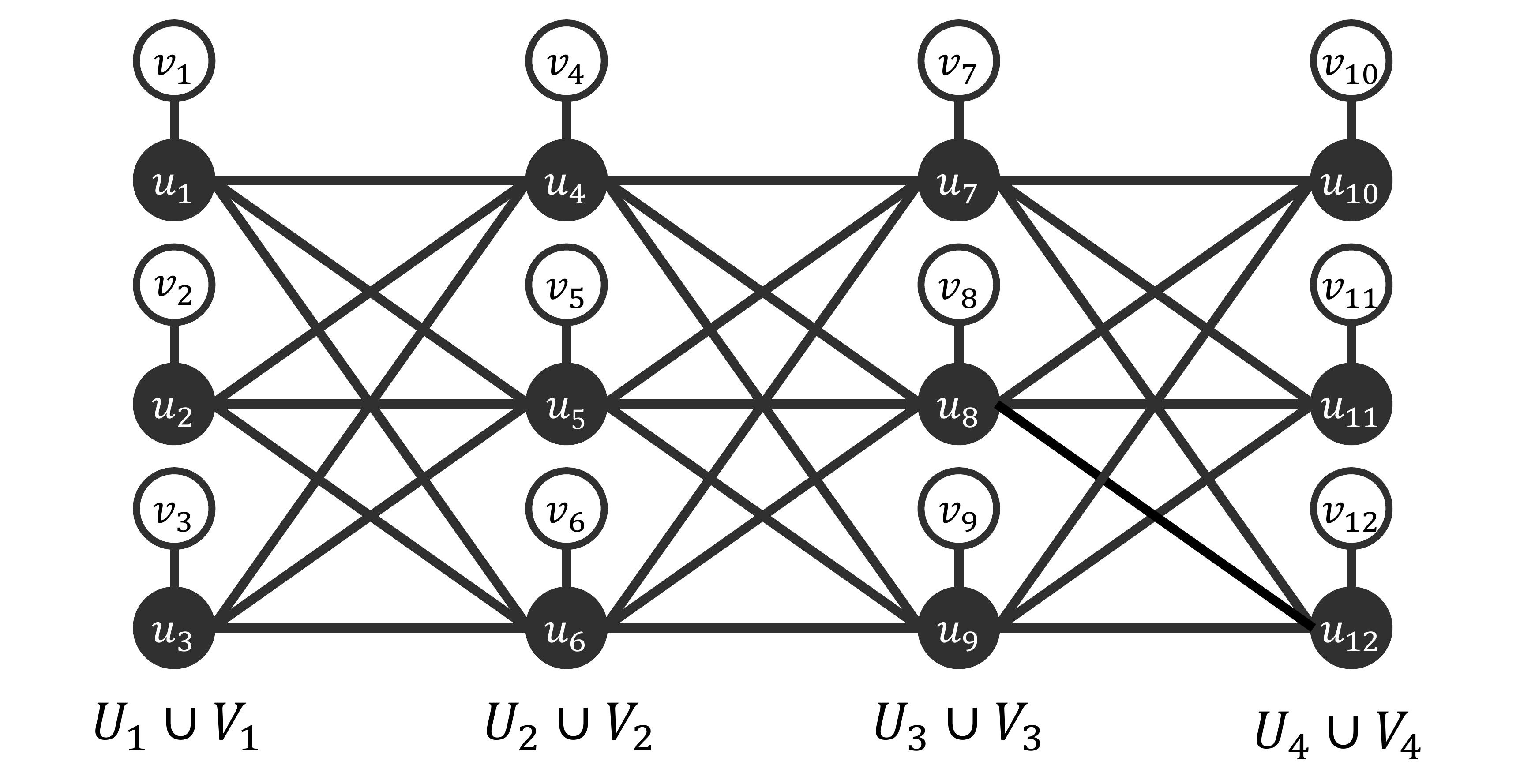}
		\caption{Hard instance for \ranking: an illustrating example with $k=3$ and $h=4$: vertices in $U$ and $V$ are represented by the solid black circles and white circles, respectively.}\label{fig:rhard}
	\end{figure}
	
	Let the deadline of vertex $u_1$ be reached first, then $u_2$'s deadline, $u_3$'s deadline, etc.
	
	\paragraph{Competitive Ratio of \ranking.}
	It is easy to see that graph $G$ has a perfect matching, by matching $u_i$ to $v_i$ for each $i\in[n]$.
	Recall that in the \ranking algorithm, each vertex $u$ is assigned a random rank $y_u\in [0,1)$.
	At the deadline of an unmatched vertex, it is matched to its unmatched neighbor $v$ (if any) with the smallest $y_v$.
	Observe that in our instance, all vertices from $U$ will be matched eventually, while each $v_i\in V$ will be matched only if at the deadline of $u_i$, $u_i$ is unmatched and $y_{v_i}$ is smaller than the ranks of all unmatched vertices from the next group $U_{\lceil\frac{i}{k}\rceil +1}$.
	
	For all $i\in[h]$, let $X_i \in \{0,1,\cdots,k\}$ be the number of unmatched vertices in $U_i$ right before the deadline of the first vertex $u_{(i-1)k+1}$ in $U_i$ is reached.
	It is easy to see that $X_{i+1}$ is a random variable that depends only on $X_i$. Hence the sequence $X_1,X_2,\ldots,X_h$ forms a Markov chain (with $k+1$ states) with initial state $X_1 = k$.
	Observe that all vertices in $U$ are matched, and the number of vertices matched in $V_i$ equals $X_i + X_{i+1} - k$. 
	Hence the competitive ratio of \ranking $\approx \frac{\sum_{i\in[h]} X_i}{kh}$.

	We say phase~$i$ begin when the deadline of the first vertex of $U_i$ is reached, and end after the deadline of the last vertex of $U_i$.
	Fix any phase~$i$, where $i<h$.
	Recall that initially $X_i$ vertices of $U_i$ are unmatched.
	Let $Z(t)$ be the number of unmatched vertices in $U_{i+1}$, when the deadlines of exactly $t$ unmatched vertices in $U_i$ have passed.
	We have $Z(0)=k$ and $X_{i+1} = Z(X_i)$.
	Let $y_1 \le y_2 \le \cdots \le y_k$ be the ranks of vertices in $U_{i+1}$.
	It is easy to see that $\expect{}{Z(t+1)} = Z(t) - 1 + y_{Z(t)}$. Taking expectation over all $y_i$'s, we have 
	$\expect{}{Z(t+1)} = \expect{}{Z(t)} - 1 + \frac{\expect{}{Z(t)}}{k+1}$.
	Let $z(\frac{t}{k}) \eqdef \frac{Z(t)}{k}$.
	It is easy to see when $k \to \infty$, $z(\frac{t}{k}) \to e^{-\frac{t}{k}}$.
	This is saying, given $X_i$, $\expect{}{X_{i+1}} = k \cdot e^{-\frac{X_i}{k}}$ when $k$ tends to infinity.
	
	Finally, note that all vertices in $U_{i+1}$ are symmetric. Hence, each of them is unmatched at the end of phase $i$ with probability $\frac{\expect{}{X_{i+1}}}{k} = e^{-\frac{X_i}{k}}$.
	Moreover, for any two vertices $u_a, u_b \in U_{i+1}$, the probability that $u_a$ is unmatched at the end of phase~$i$ is negatively correlated with the probability of $u_b$: conditioned on $u_b$ being unmatched at the end of phase~$i$, the probability of $u_a$ being unmatched is smaller.
	Thus we have measure concentration bound on $X_{i+1}$, by standard argument using moment generation function.
	In other words, the stationary distribution (when $k$ and $h$ tends to infinity) converges to a single point mass with,
	\begin{equation*}
		\frac{X}{k}=e^{-\frac{X}{k}},
	\end{equation*}
	which implies $\frac{X}{k} \approx 0.56714$, the Omega constant, which is also the competitive ratio of \ranking.
\end{proofof}

\end{document}